\documentclass[a4paper,11.999pt]{article}

\usepackage{amssymb,amsmath,amsthm,color}
\usepackage{graphicx}
\newtheorem{theorem}{Theorem}[section]
\newtheorem{cor}[theorem]{Corollary}

\begin{document}

\title{Estimating Copula and Test of Independence based on a generalized framework of all rank-based Statistics in Bivariate Sample}

\author{Abhik Ghosh    :  Roll No- MB1001\\
			 Aritra Chakravorty   : Roll No- MB1005\\
			 M.Stat. $1^{st}$ Year\\
			 Indian Statistical Institute, Kolkata}

\date{2nd June, 2011}
\maketitle

\vspace{0.5cm}

\tableofcontents\newpage

\vspace{1cm} 

\section{Introduction}
\paragraph{}
Copulas are mathematical objects that fully capture the dependence structure among random variables and hence, offer a great flexibility in building multivariate stochastic models. Since their introduction in the early 50s, copulas have gained a lot of popularity in several fields of applied mathematics, like finance, insurance and reliability theory. Nowadays, they represent a well-recognized tool for market and credit models, aggregation of risks, portfolio selection, etc.
\paragraph{}
In statistics, a copula is used as a general way of formulating a multivariate distribution in such a way that various general types of dependence can be represented. The approach to formulating a multivariate distribution using a copula is based on the idea that a simple transformation can be made of each marginal variable in such a way that each transformed marginal variable has a uniform distribution. Once this is done, the dependence structure can be expressed as a multivariate distribution on the obtained uniforms, and a copula is precisely a multivariate distribution on marginally uniform random variables. When applied in a practical context, the above transformations might be fitted as an initial step for each marginal distribution, or the parameters of the transformations might be fitted jointly with those of the copula. 
\paragraph{}
In case of bivariate sample, the notion of estimating copula is closely related to that of testing independence in a bivariate sample, as when the components of the bivariate sample are independent the copula becomes simply product of two uniform distributions. So apart from non-parametric estimation of copulas we also considered it relevant to introduce some non-parametric tests to better understand the very essence of copula in the explanation of association between the components. In fact we will develop a general multivariate statistics that gives rise to a much larger class of non-parametric rank based statistics. This class of statistics can be used in estimation and testing for the association present in the bivariate sample. We choose some representative statistics from that class and compared their power in testing independence using simulation as an attempt to choose the best candidate in that class. 
\paragraph{}
In section $\#2$ we introduce the general notion of copula, it's definition, properties and a brief remark on so called Skalr's Theorem which establishes the fact that for a given multivariate distribution copula is well defined. We discuss two methods of estimating copula in section $\#3$. In $\#3.1$ we discuss about the very intuitive empirical copula density. A natural nonparametric function that captures the dependence between two random variables is the copula, which contains all of the information which couples the two marginal distributions together to give the joint distribution of X and Y is stated in $\#3.2$.Section $\#4$ shows how copula can be used to measure the association between two elements of the bivariate sample. In section $\#5$ we introduce the concept of Rank Permutation Matrix (RPM) 
and Rank Position Vector (RPV), it's definition, properties and surprising relationship with general rank based statistics, which forms the central idea of this report. In section $\#6$ we choose some natural rank based statistics from most intuitive aspects (e.g. trace of RPM[$\#6.1$], moments of RPV[$\#6.2$],discrete fft[$\#6.3$] and mfcc[$\#6.4$] transformations of RPV)to test for independence in bivariate sample, their performance as a test statistics based on increasing sample sizes and critical values which consists of the greater bulk of this report. Again section $\#7$ is used to compare the relative performances of these statistics.  Finally section $\#8$ concludes with a discussion.
 
\vspace{1cm} 
 
\section{Copula and some properties}
\paragraph{}
The copula function provides a means to examine the dependence structure between two random variables. As defined by Schweizer and Sklar (1983) the copula, $C$, of $X$ and $Y$ is found by making marginal probability integral transforms on $X$ and $Y$ so that $C(y_1, y_2) = H\{F^{-1}(y_1), G^{-1}(y_2)\}$, $y_i \in [0, 1]~ (i = 1, 2)$, where $F$, $G$ and $H$ are the marginal and joint distribution functions of $X$, $Y$ and $(X, Y)$, respectively, and $F^{-1}$ and $G^{-1}$ are the right-continuous inverses of $F$ and $G$. Note that $C(.,.)$ is itself a bivariate distribution function on the unit square with uniform margins. We denote the corresponding probability measure by $\mu_c$. Under independence the copula is $C_I(y_1, y_2) = y_1y_2$, and any copula must fall between $max (y_1 + y_2 - 1, 0)$ and $min (y_l, y_2)$, the copulas of the upper and lower Frechet bounds. 
\paragraph{}
That the copula captures the basic dependence structure between $X$ and $Y$ is seen by the fact that all nonparametric measures of association, such as Kendall's $\tau$, Spearman's $\rho$, etc., are normed distances of the copula of $X$ and $Y$ from the independence copula. For additional discussion of the copula and its properties see Johnson (1987), Genest \& MacKay (1986) and Genest (1987).

\subsection{The Basic Idea}
\paragraph{}
Consider two random variables $X$ and $Y$, with continuous cumulative distribution functions $F_X$ and $F_Y$. The probability integral transform can be applied separately to the two random variables to define $U = F_X(X)$ and $V = F_Y(Y)$. It follows that $U$ and $V$ both have uniform distributions but are, in general, dependent if $X$ and $Y$ were already dependent (of course, if $X$ and $Y$ were independent, $U$ and $V$ remain independent). Since the transforms are invertible, specifying the dependence between $X$ and $Y$ is, in a way, the same as specifying dependence between $U$ and $V$. With $U$ and $V$ being uniform random variables, the problem reduces to specifying a bivariate distribution between two uniforms, that is, a copula. So the idea is to simplify the problem by removing consideration of many different marginal distributions by transforming the marginal variates to uniforms, and then specifying dependence as a multivariate distribution on the uniforms.

\subsection{Definition}
\paragraph{}
A copula is a multivariate joint distribution defined on the $n$-dimensional unit cube $[0, 1]^n$ such that every marginal distribution is uniform on the interval $[0, 1]$.
Specifically, $C:[0,1]^n \longmapsto [0,1]$ is an $n$-dimensional copula (briefly, $n$-copula) if: \\
\begin{itemize}
\item
$C(u)=0$ whenever $u\in [0,1]^n$ has at least one component equal to $0$;
\item
$C(u) = u_i$ whenever $u\in [0,1]^n$ has all the components equal to $1$ except the $i$th one, which is equal to $u_i$;
\item
$C$ is $n$-increasing, i.e., for each hyper-rectangle 
  \begin{equation}
  B = \times^n_{i=1}[x_i,y_i]\subseteq [0,1]^n;
  \end{equation}
  \begin{equation}
  V_C(B) = \sum_{z\in \times^n_{i=1}\{x_i,y_i\}}(-1)^{N(z)}C(z)\geq 0;
  \end{equation}
  where the $N(z) = card\{k|z_k=x_k\}$.$V_C(B)$ is the so called $C$-volume of $B$.
  \end{itemize}
  
\subsection{Sklar's theorem}
\paragraph{}  
The theorem proposed by Abe Sklar in 1959 underlies most applications of the copula.  The following is the general form of the Sklar's theorem: 
\begin{theorem}
	Given a joint distribution function $H$ for $p$ variables, and respective marginal distribution functions, there exists a copula $C$ such that the copula binds the margins to give the joint distribution.
\end{theorem}
\paragraph{}
For the bivariate case, Sklar's theorem can be stated as follows:
 \begin{theorem}
	For any bivariate distribution function $H(x,y)$, let $F(x) = H(x,\infty )$ and $G(y) = H(\infty ,y)$ be the univariate marginal probability distribution functions. Then there exists a copula $C$ such that 
\begin{equation}
H(x,y)=C(F(x),G(y))\\
\end{equation}
(where the symbol C for the copula has also been used for with its cumulative distribution function). Moreover, if the marginal distributions $F(x)$ and $G(y)$ are continuous, the copula function $C$ is unique. Otherwise, the copula $C$ is unique on the range of values of the marginal distributions.
\end{theorem}
\paragraph{}
To understand the density function of the coupled random variable $Y_H$ it should be noticed that
\begin{equation}
P[Y_H\in [x,x+dx]\times [y,y+dy]]=H(x+dx,y+dy)-H(x+dx,y)-H(x,y+dy)+H(x,y).
\end{equation}
\paragraph{}
The expectation of a function g can be written in the following ways:
\begin{equation}
E(g(X,Y))=\int\int g(x,y)dH(x,y)=\int\int g(F_X^{-1}(x),F_Y^{-1}(y))dC(x,y).
\end{equation}
\begin{equation}
E(g(X,Y))=\int^1_0\int^1_0 g(F_X^{-1}(x),F_Y^{-1}(y))\frac{\partial}{\partial x}\frac{\partial}{\partial y}C(x,y)d(x,y).
\end{equation}

\vspace{1cm} 

\section{Estimation of Copula in Bivariate sample}

	There are several methods for estimating the copula in the bivariate case. We here discuss two important methods from them.

\subsection{Empirical Copula}

 \paragraph{}
	The simplest method for estimating a copula is the Empirical copula or sometimes called the sample copula. This is defined in terms of the order statistics of the whole sample just like the empirical distribution function. Let $\{(x_k,y_k)\}_{k=1}^n$ denote a sample of size $n$ from a continuous bivariate distribution. The empirical copula is the function $C_n$ given by:
  \begin{equation}
     C_n(\frac{i}{n},\frac{j}{n})=\frac{number\ of\ pairs\ (x,y)\ in\ the\ sample\ with\ x \leq x_{(i)},y \leq y_{(j)}}{n}.\\
  \end{equation}
where $x_{(i)}$ and $y_{(j)}$, $1 \leq i,j \leq n$, denote order statistics from the sample.
\paragraph{}
We also define the empirical copula frequency $c_n$ by:
  \begin{equation}
     c_n(\frac{i}{n},\frac{j}{n})=\begin{cases}
     																\frac{1}{n} 	& \mbox{if }(x_{(i)},y_{(j)})\mbox{ is an element of the sample,}\\
     																0							&  \mbox{otherwise.}
     																\end{cases}\\
  \end{equation}
\paragraph{}
Note that $C_n$ and $c_n$ are related via the equation 
  \begin{equation}
     C_n(\frac{i}{n},\frac{j}{n})=\sum_{p=1}^{n}\sum_{q=1}^{n}c_n(\frac{i}{n},\frac{j}{n}).\\
  \end{equation}  
  and   
  \begin{equation}
     c_n(\frac{i}{n},\frac{j}{n})=C_n(\frac{i}{n},\frac{j}{n})-C_n(\frac{i-1}{n},\frac{j}{n})-C_n(\frac{i}{n},\frac{j-1}{n})+C_n(\frac{i-1}{n},\frac{j-1}{n}).\\
  \end{equation}
  Empirical copulas were introduced and first studied by Deheuvels (1979), who called them empirical dependence functions. Empirical copulas can also be used to construct nonparametric tests for independence [Deheuvels 1979, 1981a,b] .

\subsection{The Copula-Graphic estimator}
\paragraph{}
  The previous approach of empirical copula is simple but it uses the full data-set. However in the competing-risk framework we generally observe only$ T_i = min (X_i, Y_i) $ and $\delta_i = Ind(X_i < Y_i) ~~~ \forall i = 1 ,2 ,\cdots , n $. In this section we develop a estimator, called the Copula-Graphic estimator, which uses only such type of partial data --- not the full data. For we suppose that $ Pr (X_i = Y_i) = 0 $. Then using these data we can directly estimate 

  \begin{equation}
   		k(t) = Pr(X>t, Y>t),   
   		\end{equation}
   		\begin{equation}
     p_1(t) = Pr(X\leq t, X < Y),
     \end{equation}
     \begin{equation}
      p_2(t) = Pr(Y \leq t,Y < X), ~(0\leq t \leq \inf).\\
  \end{equation}

  It is well known that, under the assumption of independence of $X$ and $Y$, the marginal distribution of $X$ is uniquely determined by these probabilities. We now show the more general result that, if the copula of $(X, Y)$ is known, then the marginal distributions of both $X$ and $Y$ are uniquely determined by the competing risk data. The following theorem is relevant in this context.\\
  \begin{theorem}
     Suppose the marginal distribution functions of $(X, Y)$ are continuous and strictly increasing in $(0,\infty)$. Suppose the copula, $C$, of $(X, Y)$, is known, and $\mu(E) > 0$ for any open set $E$ in $[0, 1] \times [0, 1]$. Then $F$ and $G$, the marginal distribution functions of $X$ and $Y$, are uniquely determined by ${k(t), p_1(t), p_2(t), t > 0}$.
  \end{theorem}

  \begin{cor}
   Let $u(x, y)$ be the density function of $C$. If $u(x, y) > 0$ for any $(x, y) \in [0, 1] \times [0, 1]$, then the result of Theorem 3.1 holds. 
  \end{cor}

   Sometimes the condition that $F$ and $G$ are strictly increasing on $(0,\inf)$ is not satisfied since there may exist a time, $t_0$, such that $F(t) = 1$ for $t > t_0$. To deal with this situation, we can prove the following corollary: \\
  \begin{cor} In Theorem 3.1, if there are times $t_1$ and $t_2$ such that $F(t_1) = 1$ and $G(t_2) = 1$, and both $F$ and $G$ are strictly increasing in $(0, t_1)$ and $(0, t_2)$, respectively, then $F$ and $G$ are uniquely determined on $(0, min (t_1, t_2))$.
  \end{cor}

  \paragraph{}
   Note that here the marginal, $F$ and $G$, are only uniquely determined up to $min(t_1, t_2)$, which is reasonable since no data are observed after time $min(t_1, t_2)$.
   We have seen in the previous discussion that, given the copula, the marginal distributions of $X$ and $Y$ are uniquely determined by estimable quantities. The next step is to estimate these distributions, given the observed data and the assumed copula. 
   If $F$ and $G$ are the marginal distributions of $X$ and $Y$, then for any $t$, we have
   
  \begin{equation}
    \mu_c(A_t)=Pr(X>t,Y>t)=k(t) , 
      \end{equation}
  \begin{equation}
    \mu_c(B_t)=Pr(X\leq t,Y>t)=p_1(t) , 
   \end{equation}
   
   where
   \begin{equation}
    A_t=\{(x,y)|,F(t)<x\leq 1,G(t)<y\leq 1\},B_t=\{(x,y)|,0\leq x\leq F(t),GF^{-1}(t)\leq y\leq 1\}, 
      \end{equation}
  \begin{equation}
    \mu_c(B_t)=Pr(X\leq t,Y>t)=p_1(t) , \\
  \end{equation}

 \begin{figure}[h]
	\centering
		\includegraphics[width=79mm, height=60mm]{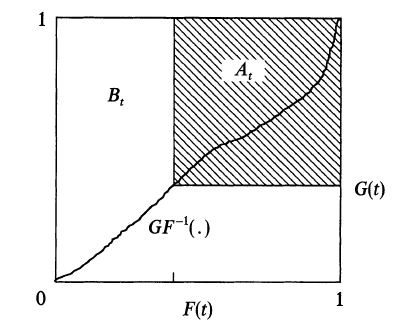}
	\caption{Relation of F(t) and G(t) on the unit square}
	\label{fig:cge}
\end{figure}

   From the proof of Theorem 3.1 we see that these two relationships uniquely determine $F$ and $G$. We find estimators $\hat{F}$ and $\hat{G}$ of $F$ and $G$ which preserve these properties on a selected grid of $m$ points $0<t_1 <t_2<\cdots < t_m < max {T_i, i = 1,\cdots, n}$. To construct our estimator, let
 $$
    \hat{A}_t=\{(x,y)|,\hat{F}(t)<x\leq 1,\hat{G}(t)<y\leq 1\},
    $$
    $$
    \hat{B}_t=\{(x,y)|,0\leq x\leq F(t),GF^{-1}(t)\leq y\leq 1\}, 
    $$
    $$
    \mu_c(B_t)=Pr(X\leq t,Y>t)=p_1(t) , \\
 $$  
   
   Let $Pr_{est}(X>t,Y>t)=n^-1\sum I(T_j>t) $, $Pr_{est}(X\leq t,X<Y)=n^-1\sum I(T_j\leq t,\delta _j=1)$ be the empirical estimates of $Pr(X>t,Y>t)$ and $Pr(X\leq t,X<Y)$.In fact any consistent estimators of $Pr(X>t,Y>t)$ and $Pr(X\leq t,X<Y)$ can be used. Find $\hat{F}(t_i)$ as the root of (4 4), subject to the definition of $\hat{G}(t_i)$ as a function of $\hat{F}(t_i)$ by solving equation(4 3). Let $\hat{F}$ and $\hat{G}$ be straight lines in each interval $(t_i, t_{i + 1})$.
   
   \begin{equation}
    \mu_c(\hat{A}_{t_i})=Pr_{est}(X>t_i,Y>t_i)=0 , \\
    \mu_c(\hat{B}_{t_i})=Pr_{est}(X\leq t_i,Y>t_i)=0 , \\
   \end{equation}

   The algorithm for constructing $F$ and $G$, based on a bisection root-finding algorithm, is as follows:  \\
\begin{itemize}
\item
   Step 1: For $i = 1$, given an initial guess for $\hat{F}(t_1)$, say $\hat{F}^{(1)}(t_l)$, find $\hat{G}^{(1)}(t_1)$ by solving (4.3).
\item
   Step 2. If this $(\hat{F}(t_1),\hat{G}(t_1))$ satisfies (4 4), go to Step 4. \\
   Otherwise use (4 4) to decide whether the next guess, $\hat{F}(t_2)$, is larger or smaller than this $\hat{F}(t_1)$. \\
   Use the midpoint of the interval $(\hat{F}(t_1), 1)$ or $(0, \hat{F}(t_1)$ as the value of $\hat{F}(t_2)$ accordingly. 
\item   
   Step 3. Repeat Steps 1 and 2 using the current estimate of $\hat{F}(t)$ and $\hat{G}(t)$. After i steps the new value of $\hat{F}(t)$ is the midpoint of either the interval $(a_i, \hat{F}^{(i)}(t_i))$ or $(\hat{F}^{(i)}(t_i), b_i)$, where, \\
   \begin{equation}
    a_i= max \{\hat{F}^{(k)}(t_1)|k<i,\hat{F}^{(k)}(t_1)<\hat{F}^{(i)}(t_i) , \\
    b_i= min \{\hat{F}^{(k)}(t_1)|k<i,\hat{F}^{(k)}(t_1)>\hat{F}^{(i)}(t_i) , \\
   \end{equation}
    This process continues until we find $\hat{F}(t_1)$ and $\hat{G}(t_1)$ which satisfy (4 3) and (4 4). It is clear that the convergence of this algorithm is guaranteed. 
\item
   Step 4. Repeat Steps 1-3 for $i = 2,\cdots, m$. Any increasing right-continuous function such that the function $\hat{G}\hat{F}^{-1}(t)$ is a straight line on each interval $[\hat{F}(t_{i-1}),\hat{F}(t_i)]$ will yield a consistent estimator. For $t> max \{T_i\}$, define $\hat{F}(t)$ and $\hat{G}(t)$ to be $\hat{F}(t_m)$ and $\hat{G}(t_m)$ accordingly.
 \end{itemize}
   
   The following result is proved in the 1992 Ohio State University Ph.D. dissertation 'On the use of copulas in dependent competing risks' by Ming   Zheng, page 47.
   
   \begin{theorem}
   Suppose that two marginal distribution functions $F$, $G$, are continuous and strictly increasing on $(0,\infty )$, and the assumed copula has density function $u(x, y) > 0$on $[0, 1] \times [0, 1]$. Then $\hat{F}_n$ and $\hat{G}_n$ are strongly consistent for $F$ and $G$. That is with probability $1$ as $n\longrightarrow \infty$, $\hat{F}_n(t)\longrightarrow F(t)$ and $\hat{G}_n(t)\longrightarrow G(t)$ for all $t \in [0, \infty )$. 
   \end{theorem}
   
   \paragraph{}
Using Theorem 3.1 and the proof of Theorem 3.4 one can show that given any continuous copula for $X$ and $Y$ and set of estimable probabilities $k(t) = Pr(X > t, Y> t)$, $p_1(t) = Pr(X\leq t, X < Y)$ and $p_2(t) = pr (Y \leq t, Y < X)$ there exists at least one set of marginal distributions $F$ and $G$ which make this possible. This result suggests that combining the observable data with any assumed continuous copula will yield an estimate of a well- defined marginal survival function. \\
\paragraph{}   
   A natural way of choosing the grid on which the above estimator is calculated is to take $t_1,\cdots, t_m$. to be the distinct times at which individuals die or are censored. Here $m$ is the number of such distinct times. With this grid and using a step function in each interval $(t_i-1, t_i)$, the estimator defined above is much easier to compute. We shall call this estimator the copula-graphic estimator. For this estimator, if $\delta_i = 1$ then $\hat{G}(t_i) = \hat{G}(t_i-1)$, while if $\delta_i = 0$ then $\hat{F}(t_i) = \hat{F}(t_i-1)$. Let $t_0$ be $0$, and $F(t_0) = G(t_0) = 0$. We have that,\\
  for $\delta_i = 1$,
  \begin{equation}
   \mu_c(A_{t_i}) = 1 - \hat{F}(t_i)- \hat{G}(t_{i-1}) + C\{\hat{G}(t_i), \hat{G}(t_{i-1})\} = Pr_{est} (X > t_i, Y> t_i); 
  \end{equation}
    and for $\delta_i = 0$,
  \begin{equation}
   \mu_c(A_{t_i}) = 1 - \hat{F}(t_{i-1})- \hat{G}(t_i) + C\{\hat{G}(t_{i-1}), \hat{G}(t_i)\} = Pr_{est} (X > t_i, Y> t_i). 
  \end{equation}  
    Then $F(t_i)$ and $G(t_i)$ are found by solving equation (20) or (21) iteratively. 
    \\
  \begin{theorem}
    The copula-graphic estimator is a maximum likelihood estimator.\\
 \end{theorem}

The following theorem, proved on page 57 of Zheng's Ph.D. dissertation, shows that the new estimator can be considered a generalisation of the Kaplan-Meier estimator to non-independent censoring.\\

  \begin{theorem}
    For the independence copula $C(x, y) = xy$, when $t < t_n$, the largest observed time, the copula-graphic estimates of marginal survival functions are exactly the Kaplan- Meier estimates.\\
  \end{theorem}
  
\paragraph{}    
    To estimate the variance of our estimator, we use the jack-knife variance estimator  \cite{efron}. That is,:\\
  \begin{equation}
    Var_{s_{est}} = \frac{n-1}{n}\sum^n_{i=1}\{\hat{S}_{(i)}(t)-\hat{S}_{(.)}(t)\}.
  \end{equation}  
    where $\hat{S}_{(i)}$ is the copula-graphic estimator using \\$(T_1,\delta_1),\cdots,(T_{i-1},\delta_{i-1}),(T_{i+1},\delta{i+1}),\cdots,(T_n,\delta_n)$
  \begin{equation}
    \hat{S}{(.)} = \frac{1}{n}\sum^n_{i=1}\hat{S}{(i)} .
  \end{equation}  
Simulation results given in \cite{zheng} shows that this estimator of variance performs reasonably well.

   
\section{Some well-known copula-based measures of association and their estimates in bivariate sample}
\paragraph{}
	There are several well-known measure of association in the bivariate sample that are some functional of the Copula. We here consider some of them. The population versions of Spearman's $\rho$, Kendall's $\tau$, and Gini's $\gamma$ are some such measures which are related to the copula by the following equations. For continuous random variables $X$ and $Y$ with copula $C$ we have
  \begin{equation}
     \rho=12\int\int_{I^2}[C(u,v)-uv]dudv,\\
  \end{equation}
  \begin{equation}   
     \tau=2\int^{1}_{0}\int^{1}_{0}\int^{v'}_{0}\int^{u'}_{0}[c(u,v)c(u',v')-c(u,v')c(u',v)]dudvdu'dv',\\
  \end{equation}
  \begin{equation}
     \gamma=4[\int^{1}_{0}C(u,1-u)du-\int^{1}_{0}[u-C(u,u)]du.\\
  \end{equation}  
   
    We can easily estimates above measures by their corresponding sample versions. In the next theorem, we present the corresponding estimators for a sample of size $n$ (we use Latin letters for the sample statistic ):\\
  \begin{theorem}
	 Let $C_n$ and $c_n$ denote the empirical copula and empirical copula frequency function for the sample $\{ (x_k,y_k) \}^n_k=1$. If $\hat{\rho}$, $t$ and $g$ denote, respectively the sample versions of Spearman's rho, Kendall's tau and Gini's gamma, then:
  \begin{equation}
     \hat{\rho}=\frac{12}{n^2-1}\sum^{n}_{i=1}\sum^{n}_{j=1}[C_n(\frac{i}{n},\frac{j}{n})-\frac{i}{n}.\frac{j}{n}].
  \end{equation}
  \begin{equation} t=\frac{2n}{n-1}\sum^{n}_{i=2}\sum^{i-1}_{p=1}\sum^{j-1}_{q=1}\sum^{n}_{j=2}[c_n(\frac{i}{n},\frac{j}{n})c_n(\frac{p}{n},\frac{q}{n})-c_n(\frac{i}{n},\frac{q}{n})c_n(\frac{j}{n},\frac{p}{n})].
  \end{equation}
  and
  \begin{equation}
   g=\frac{2n}{\left\lfloor n^2/2\right\rfloor}\{\sum^{n-1}_{i=1}[C_n(\frac{i}{n},1-\frac{i}{n})-\sum^{n}_{i=1}[\frac{i}{n}-C_n(\frac{i}{n},\frac{i}{n})].
  \end{equation}  
  \end{theorem}
  
\vspace{1cm}  

\section{A general approach for all rank based Test statistics : The RPM and the RPV}  
	Now we develop a general rank-based statistics "Rank-Position vector" (RPV) from which we can derive as a special class of functional all the rank-based measures of association including the above copula based measures also. Thus the study of this general statistic $RPV$ only is sufficient to cover all the rank-based statistics used in the inference. We will also prove the minimal sufficiency of our general statistics $RPV$. Let us first start with defining the Rank-based Permutation Matrix (RPM):

\subsection{The Rank-based Permutation Matrix (RPM)} 
	We consider the rank of the components in the bivariate sample and introduce the notion of the Rank-based Permutation Matrix (RPM) as follows:
	\begin{itemize}
\item
	Consider a bivariate sample of $n$ observations $ (X_1,Y_1), (X_2,Y_2), \cdots,(X_n,Y_n) $.
\item
	Order the $X$ and $Y$ components of the observations separately $ X_{(1)} \leq X_{(2)} \leq \cdots \leq X_{(n)}$ and 
	 $ Y_{(1)} \leq Y_{(2)} \leq \cdots \leq Y_{(n)}$. 
\item
	Assume $X_{t_i} = X_{(i)}$ and $Y_{t_i} = Y_{(s_i)}  \forall i $ .
\item
	Then the Rank-based Permutation Matrix (RPM) $R$ is defined as the $n \times n$ square     binary matrix with elements given by

\begin{eqnarray}
r_{ij} &= ~ 1 ~~~~~~~~~~ \mbox{if} ~~~ j = s_i  \\
     &= ~ 0 ~~~~~~~~~~ \mbox{otherwise}
\end{eqnarray}
\end{itemize}

\paragraph{}
Note that the $RPM$ is a permutation matrix. Under the perfect positive association the $RPM$ will be the Identity matrix and under the perfect negative association it will be the permutation matrix with all the back-diagonal entries being one :
$$
RPM~(under ~perfect ~positive ~association) = 
\begin{pmatrix}
	1 & 0 & 0 & \cdots & 0 \\
	0 & 1 & 0 & \cdots & 0 \\	
	0 & 0 & 1 & \cdots & 0 \\	
	\vdots & \vdots & \vdots & \ddots & \vdots \\
	0 & 0 & 0 & \cdots & 1 \\
\end{pmatrix}
$$ 
and
$$
RPM~(under ~perfect ~negative ~association) = 
\begin{pmatrix}
	0 & 0 & \cdots & 0 & 1\\
	0 & 0 & \cdots & 1 & 0 \\	
  \vdots & \vdots &  \ddots &\vdots & \vdots \\
	0 & 1 & \cdots & 0 & 0 \\	
	1 & 0 & \cdots & 0 & 0 \\	
\end{pmatrix}
$$ 

\paragraph{}
However under independence of $X$ and $Y$ , the sample $RPM$ can be any one of the all possible permutation matrix with equal probability. In other words the distribution of the sample $RPM$ under the assumption of independence is uniform over all possible permutation matrices. Since there are $n!$ possible permutation matrices of order $n \times n $, the corresponding probability for the sample $RPM$ to take any particular value is $\frac{1}{n!}$.And as the sample deviates from the assumption of independence the sample $RPM$ teds to exhibits certain particular patterns with higher probability and hence deviates from the uniformity of the distribution. Finally the sample distribution of $RPM$ becomes degenerate under the perfect association as indicated above. Therefore, we can use the notion of $RPM$ to discriminate the case of independence against all kinds of alternatives of non-independence.

\subsection{The Rank-Position Vector (RPV)}
 For an $ n \times n$ $RPM$, we now define a n-vector, the Rank-Position Vector ($RPV$), $\textbf{s} $ where the $i^{th}$ element of the vector will be the number of the column that have one in its $i^{th}$ row. That is, in our earlier notations 
$$
\textbf{s} = 
\begin{pmatrix}
	s_1\\
	s_2 \\	
  \vdots  \\
	s_n \\	
\end{pmatrix}
$$ 
Note that there is a one-to-one correspondence between the $RPM$ and the $RPV$ in the sense that one can be derived from the other. Similar to the $RPM$ this vector $\textbf{s}$ will also a relatively lower dimensional rank-based statistics whose distribution changes with the value of association between the two variables. From the null distribution of $RPM$ we can derive the null distribution of $RPV$ $\textbf{s}$ which is simply the uniform distribution over all possible permutations of $\{1,2,\cdots,n\}$	with each probability being $\frac{1}{n!}$. Hence the null distribution of the $RPV$ is also independent of the underlying distribution of the random variables. And as the samples deviates from independence its distribution also deviates from uniform to certain pattern finally becoming degenerate for the perfect association having value:
$$
RPV~(under ~perfect ~positive ~association) = 
\begin{pmatrix}
	1 \\
	2 \\	
	3 \\	
	\vdots \\
	n \\
\end{pmatrix}
$$ 
and
$$
RPV~(under ~perfect ~negative ~association) = 
\begin{pmatrix}
	n\\
	n-1 \\	
  \vdots  \\
	2 \\	
	1\\	
\end{pmatrix}
$$ 
\paragraph{}
	Since the null distribution of the $RPV$ is obtained we can now find the first two moments of the $RPV$ $\textbf{s}$  under the assumption of independence. Under above assumption each coordinates of $\textbf{s}$, say $s_i$ can takes values $1, 2, \cdots, n$ with probability $\frac{(n-1)!}{n!} = \frac{1}{n}$ each and hence each coordinate has expectation $\frac{n+1}{2}$. Thus
	
\begin{equation}
	E_0(RPV) = E_0(\textbf{s}) = \frac{(n+1)}{2}~\textbf{1} 
\end{equation}
where $E_0(.)$ denotes the expectation under the null hypothesis $H_0: r=0$ and $\textbf{1}$ be the $n$-vector of all entries one. Similarly the variance of any co-ordinates of the $RPV$ $\textbf{s}$ equals $\frac{n^2-1}{12}$ and the covariance between any two co-ordinates is $- \frac{n^2-1}{12(n-1)}$. Hence we get,
\begin{equation}
	V_0(RPV) = V_0(\textbf{s}) = 
	\begin{pmatrix}
	\frac{n^2-1}{12} & - \frac{n^2-1}{12(n-1)} & \cdots & - \frac{n^2-1}{12(n-1)} & - \frac{n^2-1}{12(n-1)}\\
	- \frac{n^2-1}{12(n-1)} & \frac{n^2-1}{12} & \cdots & - \frac{n^2-1}{12(n-1)} & - \frac{n^2-1}{12(n-1)} \\	
  \vdots & \vdots &  \ddots &\vdots & \vdots \\
	- \frac{n^2-1}{12(n-1)} & - \frac{n^2-1}{12(n-1)} & \cdots & \frac{n^2-1}{12} & - \frac{n^2-1}{12(n-1)} \\	
	- \frac{n^2-1}{12(n-1)} & - \frac{n^2-1}{12(n-1)} & \cdots & - \frac{n^2-1}{12(n-1)} & \frac{n^2-1}{12} \\	
\end{pmatrix}
\end{equation}
where $V_0(.)$ denotes the dispersion matrix under the null hypothesis $H_0: r=0$. Note that under independence the random vector $RPV$ has the equi-correlation structure with the value of the correlation coefficient between any two co-ordinates of the $RPV$ being  $- \frac{1}{(n-1)}$, the least possible intra-class correlation between the two random variables $(X,Y)$.

\subsection{Relationship of the RPV and the usual rank-based statistics}
\paragraph{}
It is interesting to note that any permutation invariant rank based test-statistics based on a bivariate sample can be seen as a function of the RPV. In particular if we consider the rank based measures of correlation from section 3.1 which showed closed relationship with empirical copula based on a bivariate sample, one should expect them to be expressed in terms of the $RVP$. These expressions are precisely stated in the following theorem:
\begin{theorem}
    Let $ s = ( s_1,s_2,\cdots,s_n )' $ denote the $RPV$ for the sample $\{ (x_k,y_k) \}_{k=1}^n$. If $r$, $t$ and $g$ denote, respectively the sample versions of Spearman's rho, Kendall's tau and Gini's gamma, then:
  \begin{equation}
     r = \sum^n_{k=1}k.s_k - [n(n+1)/2]^2.
  \end{equation}
  \begin{equation}
     t = 2\sum^n_{k=1}s_k - [n(n+1)].
  \end{equation}
  \begin{equation}
     g = 2\sum^n_{k=1}min\{(n+1)/2-k,(n+1)/2-s_k\}.
  \end{equation}
  \end{theorem}

  \begin{theorem}
  	Any permutation invariant rank-based statistics based on a bivariate sample \\$(X_1,Y_1), (X_2,Y_2), \cdots, (X_n,Y_n)$ will be a function of the $RPV$ of that sample.
  \end{theorem} 
\begin{proof}
   Let $R_{X_i}$ and $R_{Y_i}$ be the rank of the $i^{th}$ observation $X_i$ and $Y_i$ respectively $\forall i$ and consider any rank-based statistic $T = f((R_{X_1},R_{Y_1}),\cdots,(R_{X_n},R_{Y_n}))$ where $f(.)$ is a permutation invariant function of its arguments. Next note that from the definition of the $RPV = \textbf{s} = (s_1,s_2,\cdots,s_n)'$ it follows that [se eq.() and () above] if for any $i$ the $X$-rank $R_{X_i}=j$ then the corresponding $Y$-rank will be $R_{Y_i}= s_j$. Also as $i$ ranges over $\{1,2,\cdots,n\}$ the corresponding $X$-rank $R_{X_i}$ also ranges over ${1,2,\cdots,n}$. Thus using the permutation invariance of the function $f$, we obtain
  \begin{equation}
  T = f((R_{X_1},R_{Y_1}),\cdots,(R_{X_n},R_{Y_n})) = f((1,s_1),\cdots,(n,s_n))
  \end{equation}
  which is only a function of the statistic $(s_1,s_2,\cdots,s_n)' = RPV$
  \end{proof}
\paragraph{}
	Note that the restriction of the permutation invariance is not of much problematic because we usually consider the permutation invariant rank-based statistics only to get mare information and less variance, e.g., usual linear rank-based statistics. However, in view of the above theorem, our general rank-based statistic $RPV$ gives rise to a more general class of nonparametric rank-based statistics for the bivariate random variables including all the usual permutation invariant rank-statistics and many more rank-statistics those may be very efficient in inference. So, we can use several suitable statistics from this class of statistics in the non-parametric inference about the association between the two random variables. We already seen above that some of the usual measure of association are nothing but some simple functional of $RPV$. Also we can develop several other measures also from this class of statistics based on $RPV$. However here we will develop some useful test statistics for testing the independence in the bivariate sample.

\vspace{1cm} 

\section{Some nonparametric tests of independence in bivariate sample based on RPM and RPV}

\paragraph{}
	Now we consider the problem of testing for independence between two components of a bivariate sample. We try to develop some tests from several logical inference. There are several well known parametric test. But in all those cases we need to specify the underlying distribution of the variables. However in practice we usually do not have any idea about the underlying distribution so that it is difficult to implement those parametric tests. So we here propose some non-parametric tests. 	For this purpose we consider the Rank-based Permutation Matrix (RPM) and the Rank-Position Vector (RPV) in the bivariate sample and try to develop some test statistics for independence based on them. We already argued that the distribution of $RPM$ and $RPV$ changes with the change in the measure of association between the two random variables in the sample and so we can use them or some suitable functional of them to test for independence.
\paragraph{}
	One major advantage in the use of $RPM$ or $RPV$ in testing for independence is their non-parametric nature. As discussed above the null distributions of both the sample $RPM$ and the sample $RPV$ do not depend on the underlying distribution of the sample $( X, Y)$. Since for testing purpose we only need the null distribution of the test statistics ( to compute the critical values of the test ), the test based on $RPM$, $RPV$ or any functional of them will be independent of the underlying distribution of the sample $(X,Y)$. However, since there is no clear partial ordering in $\Re^n$ for $n>1$ and the $n \times n ~(n>1)$ $RPM$ belongs to the space isomorphic to $\Re^{n \times n}$, the comparison of two $RPM$ directly is quite difficult. Similarly, the comparison of two $RPV$ directly is quite difficult as the $n$-vector $RPV$ belongs to the space isomorphic to $\Re^n$. So we consider some one dimensional functional of $RPM$ and $RPV$ that preserve the property that it's distribution changes as the dependence between the two random observable changes from independence to perfect association. And so we can use this one dimensional functional of $RPM$ and $RPV$ to form the test of independence in the bivariate sample. Note that, as argued above, any such test will also be independent of the underlying distribution of the samples. Below we consider some such functional of $RPM$ and $RPV$ that also have some intuitive justifications.
	
\subsection{Trace of the RPM}
\paragraph{}
	A common practice in the analysis of a matrix is to use the eigenvalue decomposition of the matrix. Here also we may consider the set of $n$ eigenvalues of the $n \times n$ $RPM$ and then consider a one dimensional functional of this set of eigenvalues. The simplest choice for that is the sum of the eigenvalues which leads to the trace of the $RPM$. So we have obtain a one dimensional test statistic
	$$
	T_1 = Trace(RPM) = \mbox{ Sum of the eigenvalues of the RPM }
	$$
	Now we will be able to form a test using the statistics $T_1$ if it preserve the property that its null distribution differ significantly from the distribution under alternative. Note that there may be some alternative for which the distribution of $T_1$ under null differ significantly from that under the alternative and there may be some other alternative for which this property does not holds. In such a case the test based on $T_1$ will not be uniformly powerful, rather it will be locally powerful against the alternatives satisfying the above requirement. However, as in the case of $RPM$, the distribution of $T_1 = Trace (RPM)$ will also depends only on the value of association between the two variables $(X,Y)$ and hence it's null distribution will be independent of the underlying distribution of the random variables $(X,Y)$. So we simulate the distribution of $T_1$ for various sample sizes $n$ from bivariate normal with means $0$, variances $1$ and correlation coefficient $r$ for various values of $r$ in $[-1, 1]$. And plot the corresponding sample mean of those distributions in the following figure \ref{fig:trc}. 
	
\begin{figure}[h]
	\centering
		\includegraphics[width=79mm, height=59mm]{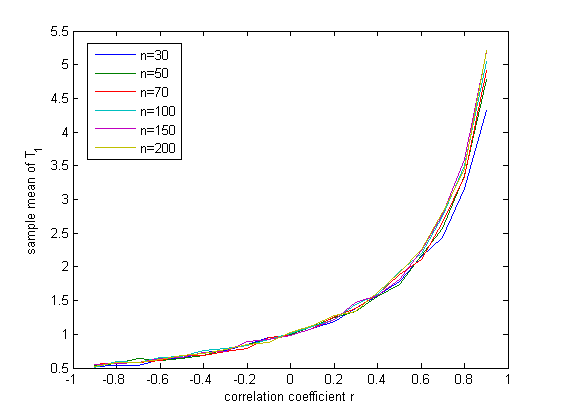}
	\caption{Plot of the sample mean of the statistics $T_1$ for various sample size $n$}
	\label{fig:trc}
\end{figure}

\paragraph{}
	From the figure \ref{fig:trc} it is clear that the distribution of $T_1$ changes as the value of correlation coefficient $r$, a measure of association changes. So we can use $T_1$ to test for independence, that is, for the null hypothesis $H_0 : r = 0$ against the alternative $H_1 : r \neq 0$. Also we note that the rate of change in the mean of $T_1$ is higher for positive values of $r$ than that for negative values of $r$. Hence we expect that the power of our test based on the trace $T_1$ will have much more power for the alternatives of positive association and it may be locally best for such alternatives. we examines the performances of this test against several alternatives in the following sectin 6 .

\paragraph{}
	Now we want to find the null distribution of $T_1$ that will be required to compute the critical values for the test of independence based on $T_1$.Though we know the null distribution of $RPM$ is known, it is quite difficult to obtain the null distribution of $T_1 = Trace(RPM)$. So we here propose to use the simulated cut-off for the test. The table \ref{tab:cv1} shows the simulated quantiles  for the null distribution of $T_1$. Using this we formulate our both sided test for independence as follows : 
	\\
	Reject $H_0 : r = 0$ ( independence ) in favour of the alternative $H_1 : r \neq 0 $ ( non-independence ) at $100\alpha$\% level of significance (e.g. $\alpha = 0.05$ )if 
	\begin{center} $T_1 \leq C1(n,\frac{\alpha}{2})$       or          $T_1 \geq C1(n,1-\frac{\alpha}{2}) $
	\end{center}
	where $ C1(n,\alpha) $ is the $100\alpha$\% quantile of the null distribution of $T_1$ for sample size $n$ that can be obtained from table \ref{tab:cv1}. Similarly we can form the one sided test for independence also.

\begin{table}[h]
	\centering
	\caption{The simulated quantiles of the sample distribution of $T_1$ under null $H_0 : r=0$ for different sample size $n$ }
		\begin{tabular}{|c| c c c c c c c c c|} \hline
		n    &  2.5\% & 5\% & 10\% &  25\%  &  50\% & 75\% &  90\% & 95\% & 97.5\% \\ \hline
	  30   &  0  &   0  &   0   &  0   &  1  &   2   &  2  &   3  &   3\\
    50   &  0  &   0  &   0   &  0   &  1  &   2   &  2  &   3  &   3\\
    70   &  0  &   0  &   0   &  0   &  1  &   2   &  2  &   3  &   3\\
   100   &  0  &   0  &   0   &  0   &  1  &   2   &  2  &   3  &   3\\
   150   &  0  &   0  &   0   &  0   &  1  &   2   &  2  &   3  &   3\\
   200   &  0  &   0  &   0   &  0   &  1  &   2   &  2  &   3  &   3\\ \hline
		\end{tabular}
	\label{tab:cv1}
\end{table}

\subsection{The first two moments of the RPV}
\paragraph{}
	We now try to develop some test based on the general statistic $RPV$. We already derived the first two moments of $RPV$ under null. So it will be intuitive to compare the sample mean of the $RPV$ with the theoretical one. However from one sample we will get only one $RPV$ and we checked that only one value of the $RPV$ is not sufficient to compare the scenario with different value of correlation. So we propose an way to get an estimate of the mean and the covariance of the $RPV$ based on only one sample. Suppose we have a sample of size $n$ $(X_1,Y_1),\cdots,(X_n,Y_n)$ . We choose an integer $k$ such that $n'=\frac{n}{k}$ is also an integer and then break the sample into $n'=\frac{n}{k}$ sub-samples each of size $k$. Now we compute the $RPV$ based on each sub-samples, say $RPV_1$,$RPV_2$,$\cdots$, $RPV_{n'}$. Note that each of the above $RPS_i$'s are computed from the $i^{th}$ subsamples of size $k$. Now we can estimate the mean and variance of the $RPV$ as 
	
\begin{equation}
	\hat{E}(RPV)= \frac{1}{n'}\sum_{i=1}^{n'}~RPV_i
\end{equation}
	and
	
\begin{equation}
	\hat{V}(RPV)= \frac{1}{n'}\sum_{i=1}^{n'}~(RPV_i-\hat{E}(RPV))(RPV_i-\hat{E}(RPV))'
\end{equation}

\paragraph{}
	Now we propose a test statistics based on the above estimated values. Note that the intuitive choice is to look for the Mahalonabis distance based on the estimated mean and variances. But we observe that this cannot distinguish the cases with different correlations so that it cannot be used to form the test for independence. The main problem in this approach is that the estimated covariance matrix becomes ill-conditioned because the original dispersion matrix of the $RPV$ is singular. However we here made an interesting observation on the estimated covariance matrix. Here we observed the diagonal entries of the estimated covariance matrix $\hat{V}(RPV)$ that is the variances of the components of $RPV$. The observation is, as the correlation increases from $-1$ to $1$ the elements of first half of the diagonal elements of $\hat{V}(RPV)$ increases and the elements of second half of the diagonal elements of $\hat{V}(RPV)$ decreases. For zero correlation all the variances becomes exactly equal confirming our previous expression of the dispersion of the variance. So if we consider the variance of the diagonal elements of $\hat{V}(RPV)$ it will be closed to zero under independence and gradually increases as the correlation moves away from $0$ to $1$ or $-1$ [fig \ref{fig:vd}]. So we can use this observation to form a test of independence using the statistics
	\begin{equation}
	 T_2 = Variance~of~the~diagonal~elements~of~\hat{V}(RPV)~.
	\end{equation}
	 
\begin{figure}[h]
	\centering
		\includegraphics[width=79mm, height=60mm]{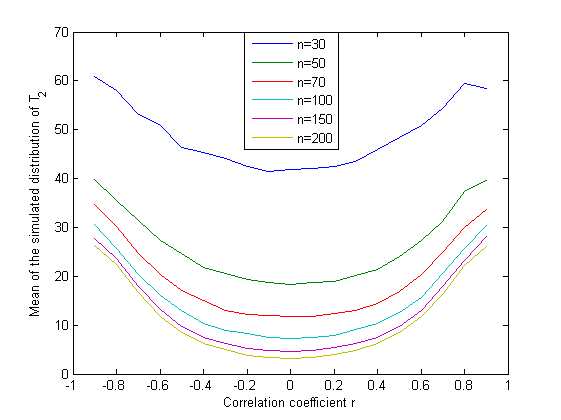}
	\caption{Plot of the sample mean of the statistics $T_2$ for various sample size $n$}
	\label{fig:vd}
\end{figure}

\paragraph{}
	Next we need to find the null distribution of $T_2$ to obtain the critical values for the test of independence based on $T_2$. However the explicit form of the distribution of $T_2$ is difficult to derive. So we here give the simulated quantiles of the null distribution of $T_2$ in the table \ref{tab:cv2}. Note that since the statistic $T_2$ is a function of $RPV$ its null distribution is independent of the underlying distribution of the random sample. Thus our both sided test for independence will be as follows :
	\\
	Reject $H_0 : r = 0$ ( independence ) in favour of the alternative $H_1 : r \neq 0 $ ( non-independence ) at $100\alpha$\% level of significance (e.g. $\alpha = 0.05$ )if 
	\begin{center} $T_2  \geq C2(n,1-\alpha) $
	\end{center}
	where $ C2(n,\alpha) $ is the $100\alpha$\% quantile of the null distribution of $T_2$ for sample size $n$ that can be obtained from table \ref{tab:cv2}.

\begin{table}[h]
	\centering
	\caption{The simulated quantiles of the sample distribution of $T_1$ under null $H_0 : r=0$ for different sample size $n$ }
		\begin{tabular}{|c| c c c c c c c c c|} \hline
		n    &  2.5\% & 5\% & 10\% &  25\%  &  50\% & 75\% &  90\% & 95\% & 97.5\% \\ \hline
		
	 30  & 13.7481  & 16.8704  & 21.1963  & 29.6296  & 40.7519  & 52.6963  & 64.3364  & 71.2889  & 77.2722\\
   50  &  6.5134  &  7.8796  &  9.6220  & 12.9756  & 17.4582  & 22.8747  & 28.5947  & 32.3536  & 35.8485\\
   70  &  3.7700  &  4.6133  &  5.7350  &  7.9646  & 10.9385  & 14.5052  & 18.2253  & 20.7343  & 23.0242\\
  100  &  2.2974  &  2.8423  &  3.5513  &  4.9707  &  6.9196  &  9.2624  & 11.7565  & 13.4325  & 15.0129\\
  150  &  1.4073  &  1.7225  &  2.1500  &  3.0254  &  4.2590  &  5.7684  &  7.3336  &  8.4031  &  9.3331\\
  200  &  1.0133  &  1.2389  &  1.5465  &  2.1758  &  3.0563  &  4.1376  &  5.3003  &  6.0786  &  6.7942 \\  \hline
		\end{tabular}
	\label{tab:cv2}
\end{table} 
	
\subsection{The Discrete Fourier Transformation of the RPV}
\paragraph{}
	Now we consider the discrete Fourier transformation of the $RPV$ and try to develop some test of independence based on that. Suppose $f_1, f_2,\cdots, f_n$ be the Discrete Fourier coefficient for the $n$-vector $RPV$ based on a bivariate sample of size $n$ where each $f_i$ possible be a complex number. Note that when the two random variables will be independent the uncertainty present in the Fourier coefficients $f_i$'s will be maximum. As the association between the two variables increases towards perfect association the uncertainty will decrease and takes its minimum value under the perfect association. So we can use some suitable measure of this uncertainty in the Fourier coefficients to test for independence. Here we will consider two common measure of uncertainty, namely the entropy of the normalized amplitude of the Fourier coefficients($T_3$) and the sum of squares of the Fourier coefficients ($T_4$). 
	\paragraph{}
	For defining these statistics, let us assume that $A_i$ denotes the amplitude (absolute value) of the $i^{th}$ Fourier coefficient $f_i$ and $A_i^*$ denotes the corresponding normalized amplitude for all $i$, i.e., $A_i^* = \frac{A_i}{\sum_{j=1}^n~A_j}$. Then we have
	
\begin{equation}
	T_3 = - \sum_{i=1}^n~~A_i^* \log(A_i^*)
\end{equation}
  and
\begin{equation}
	T_4 = \sum_{i=1}^n~~A_i
\end{equation}

These statistics being a measure of uncertainty also behave exactly similar to that explained above for the uncertainty with respect to the change in the association between the two random variables [ Figure \ref{fig:fft} ]. So we can use them to test for the null hypothesis of independence.

\begin{figure}[ht]
\centering
\includegraphics[width=79mm, height=60mm] {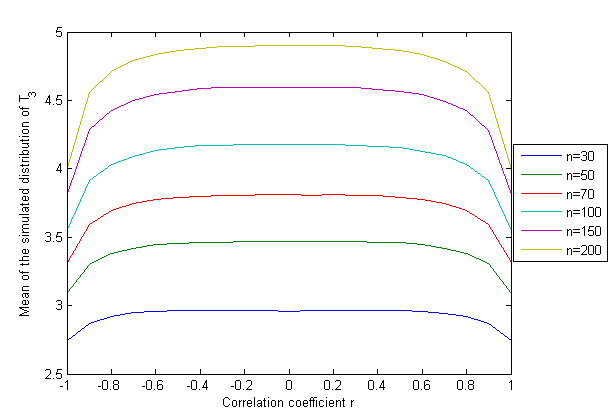}
\includegraphics[width=79mm, height=60mm] {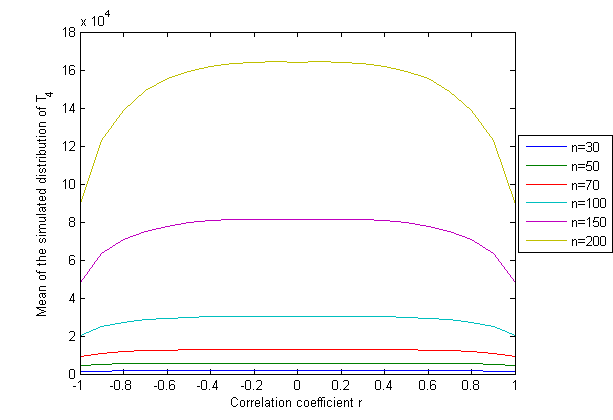}
\caption{Plot of the sample mean of the statistics $T_3$ and $T_4$ for various sample size $n$}
\label{fig:fft}
\end{figure}

 Note that the null distribution of these statistics $T_3$ and $T_4$ are independent of the underlying distribution of the random variables but the exact form of the null distribution is very difficult to find out. So in the table \ref{tab:cv3} and \ref{tab:cv4} we give the simulated quantiles of the distribution of $T_3$ and $T_4$ respectively for some sample sizes. We can use this simulated cut-off to form our test for independence against the both sided alternative as follows:
	\\
	Reject $H_0 : r = 0$ ( independence ) in favour of the alternative $H_1 : r \neq 0 $ ( non-independence ) at $100\alpha$\% level of significance (e.g. $\alpha = 0.05$ )if 
	\begin{center} $T_3  \leq C3(n,\alpha) $
	\end{center}
	where $ C3(n,\alpha) $ is the $100\alpha$\% quantile of the null distribution of $T_3$ for sample size $n$ that can be obtained from table \ref{tab:cv3}.Similarly for $T_4$ also we reject if
	\begin{center} $T_4  \leq C4(n,\alpha) $
	\end{center}
	where $ C4(n,\alpha) $ is the $100\alpha$\% quantile of the null distribution of $T_4$ for sample size $n$ obtained from table \ref{tab:cv4}.

\begin{table}[h]
	\centering
	\caption{The simulated quantiles of the sample distribution of $T_3$ under null $H_0 : r=0$ for different sample size $n$ }
		\begin{tabular}{|c| c c c c c c c c c|} \hline
		n    &  2.5\% & 5\% & 10\% &  25\%  &  50\% & 75\% &  90\% & 95\% & 97.5\% \\ \hline
		
	 30 &    2.8451  &  2.8682 &   2.8934 &   2.9314  &  2.9688  &  3.0015 &   3.0263  &  3.0395 &   3.0501\\
   50 &    3.3765  &  3.3933 &   3.4120 &   3.4412  &  3.4709  &  3.4973 &   3.5182  &  3.5298 &   3.5391\\
   70 &    3.7309  &  3.7455 &   3.7610 &   3.7851  &  3.8106  &  3.8337 &   3.8524  &  3.8629 &   3.8715\\
  100 &    4.1112  &  4.1230 &   4.1357 &   4.1559  &  4.1768  &  4.1964 &   4.2126  &  4.2218 &   4.2294\\
  150 &    4.5453  &  4.5549 &   4.5649  &  4.5813  &  4.5985  &  4.6145 &   4.6280  &  4.6357 &   4.6422\\
  200 &    4.8548  &  4.8626 &   4.8712  &  4.8851  &  4.8999  &  4.9139 &   4.9259  &  4.9328 &   4.9384 \\  \hline
		\end{tabular}
	\label{tab:cv3}
\end{table} 
	
	\begin{table}[h]
	\centering
	\caption{The simulated quantiles of the sample distribution of $T_4$ under null $H_0 : r=0$ for different sample size $n$ }
		\begin{tabular}{|c| c c c c c c c c c|} \hline
		n    &  2.5\% & 5\% & 10\% &  25\%  &  50\% & 75\% &  90\% & 95\% & 97.5\% \\ \hline
		
30 &	1568.8	&  1590.1 &	1612.6&	1647.2&	1681.8&	1712.6&	1736.7&	1749.8&	1760.2	\\
50 &	5415.7	& 5468.6 	 & 5527.6&	5621.5 &	5718.0	&  5805.0 & 	5875.6&	5915.2 &	5946.5 	\\
70 &  12268.8&	12372.4&	12485.0&	12663.3&	12853.5&	13028.8&	13172.2&	13254.2&	13322.2 	\\
100 &	29273.0&	29479.6&	29708.1&	30066.8&	30449.7 &	30807.4&	31108.0 &	31278.7&	31420.9	\\
150 &	78899.8 &	79354.3&	79854.1&	80661.1&	81516.0 &	82323.2&	83015.8&	83411.0&	83745.9 	\\
200 &	159782.1&	160565.7&	161425.6&	162853.1&	164374.1 &	165826.1&	167072.8&	167803.4&	168410.5 		\\  \hline
		\end{tabular}
	\label{tab:cv4}
\end{table}

\subsection{The Mel-frequency cepstrum coefficient of RPV}
\paragraph{}
	Instead of the Discrete Fourier Transformation of the $RPV$ as discussed in the preceding section we can also consider the Mel-frequency cepstrum coefficient of the $RPV$. And then we look at the uncertainty present in the corresponding set of the Mel-frequency cepstrum coefficients. Here also we consider the entropy and the sum of squares of the Mel-frequency cepstrum coefficients (MFCCs) of the $RPV$ to develop test for independence. First note that the MFCCs of the $RPV$ are derived as follows:
\begin{enumerate}
	\item Take the Fourier transform of (a windowed excerpt of) the $RPV$.
	\item Map the powers of the spectrum obtained above onto the mel scale, using triangular overlapping windows.
	\item Take the logs of the powers at each of the mel frequencies.
	\item Take the discrete cosine transform of the list of mel log powers, as if it were a signal.
	\item The MFCCs are the amplitudes of the resulting spectrum.
\end{enumerate}
Suppose $M_1, M_2,\cdots,M_n$ be the Mel-frequency cepstrum coefficients for the $n$-vector $RPV$ based on a bivariate sample of size $n$. Let $M_i^*$ be the normalized absolute coefficients corresponding to the $i^{th}$ MFCC $M_i$ for all $i$, i.e., $M_i^* = \frac{|M_i|}{\sum_{j=1}^n~|M_j|}$. Then the two test staistics corresponding to the entropy and sum of squares are respectivly given by 

\begin{equation}
	T_5 = - \sum_{i=1}^n~~M_i^* \log(M_i^*)
\end{equation}
  and
\begin{equation}
	T_6 = \sum_{i=1}^n~~M_i^2
\end{equation}

\begin{figure}[ht]
\centering
\includegraphics[width=79mm, height=60mm] {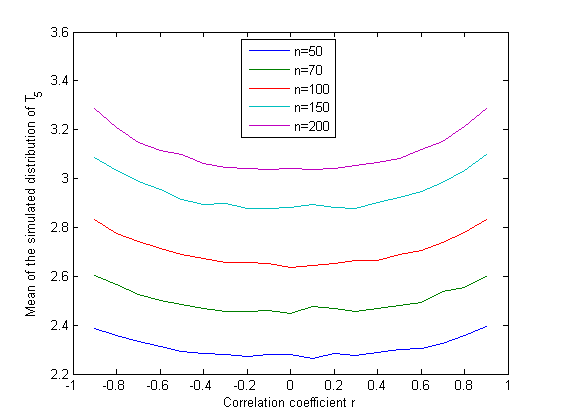}
\includegraphics[width=79mm, height=60mm] {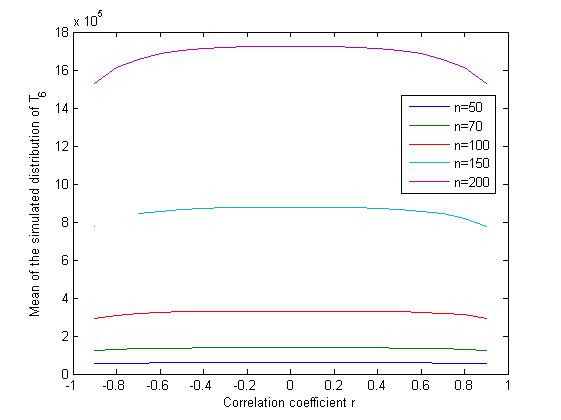}
\caption{Plot of the sample mean of the statistics $T_5$ and $T_6$ for various sample size $n$}
\label{fig:mfcc}
\end{figure}

We here made an interesting observation regarding the relationship between the change in the above test statistics $T_5$ and $T_6$ based on the MFCCs and the change in the measure of association in the bivariate sample. The statistic $T_6$ is maximum under independence like the previous case with the Fourier coefficients and gradually decreases as the sample deviates from independence. But on the contrary, in this case the entropy based statistics $T_5$ is minimum under the null hypothesis of independence and gradually increases as the association between the two variables increases [ Figure \ref{fig:mfcc} ]. However the distribution of the statistics $T_5$ and $T_6$ under the assumption of independence will also be independent of the underlying distribution of the bivariate samples. But due to the unavailability of the exact form of the null distribution of $T_5$ and $T_6$, here also we need to use the simulated cut-off given in the table \ref{tab:cv5} and \ref{tab:cv6}. Therefore the test of independence based on $T_5$ against the both sided alternative will be given by :
	\\
	Reject $H_0 : r = 0$ ( independence ) in favour of the alternative $H_1 : r \neq 0 $ ( non-independence ) at $100\alpha$\% level of significance (e.g. $\alpha = 0.05$ )if 
	\begin{center} $T_5  \geq C5(n,1-\alpha) $
	\end{center}
where $ C5(n,\alpha) $ is the $100\alpha$\% quantile of the null distribution of $T_5$ for sample size $n$ that can be obtained from table \ref{tab:cv5}. Similarly the test based on $T_6$ will reject $H_0 : r = 0$ ( independence ) against the both sided alternative if
	\begin{center} $T_6  \leq C6(n,\alpha) $
	\end{center}
where $ C6(n,\alpha) $ is the $100\alpha$\% quantile of the null distribution of $T_6$ for sample size $n$ obtained from table \ref{tab:cv6}.

\begin{table}[h]
	\centering
	\caption{The simulated quantiles of the sample distribution of $T_5$ under null $H_0 : r=0$ for different sample size $n$ }
		\begin{tabular}{|c| c c c c c c c c c|} \hline
		n    &  2.5\% & 5\% & 10\% &  25\%  &  50\% & 75\% &  90\% & 95\% & 97.5\% \\ \hline
		
	 30 &    1.6912  &  1.7329 &   1.7885 &   1.8909 &   2.0166 &   2.1524 &   2.2767 &   2.3497 &   2.4131\\
   50 &    1.9305  &  1.9809 &   2.0421 &   2.1478 &   2.2723 &   2.4024 &   2.5246 &   2.5980 &   2.6599\\
   70 &    2.1059  &  2.1592 &   2.2217 &   2.3275 &   2.4495 &   2.5756 &   2.6941 &   2.7668 &   2.8286\\
  100 &    2.3095  &  2.3607 &   2.4224 &   2.5263 &   2.6451 &   2.7673 &   2.8783 &   2.9461 &   3.0051\\
  150 &    2.5522  &  2.6018 &   2.6609 &   2.7608 &   2.8740 &   2.9913 &   3.0994 &   3.1651 &   3.2238\\
  200 &    2.7187  &  2.7669 &   2.8251 &   2.9239 &   3.0338 &   3.1464 &   3.2497 &   3.3138 &   3.3709\\  \hline
		\end{tabular}
	\label{tab:cv5}
\end{table} 
	
	\begin{table}[h]
	\centering
	\caption{The simulated quantiles of the sample distribution of $T_6$ under null $H_0 : r=0$ for different sample size $n$ }
		\begin{tabular}{|c| c c c c c c c c c|} \hline
		n    &  2.5\% & 5\% & 10\% &  25\%  &  50\% & 75\% &  90\% & 95\% & 97.5\% \\ \hline
		
30&	  14441&	 14683	&14946&	15351&	15759	&16122	&16401	&16549	&16669	\\
50& 	55739 &	56321	 &56985	  &58028&	59100&	60058	&60831	&61255	&61604	\\
70& 	131664&	132750	&133952	&135824&	137793&	139595	&141066	&141904	&142578	\\
100&	320247&	322213	&324382	&327837&	331440&	334802&	337597	&339184&	340484	\\
150	& 855835&	859596	&863666	&870382	&877419	&883981	&889521	&892672	&895337	\\
200	& 1690472&	1696362&	1702761&	1713086&	1724009&	1734375&	1743265&	1748391&	1752574		\\  \hline
		\end{tabular}
	\label{tab:cv6}
\end{table}

\vspace{1cm} 

\section{Simulation study : Comparison of the tests}
\paragraph{}
	Now in this section we will check the performances of the above tests with respect to some common alternatives and compare their relative performances. We simulate the power of the tests we proposed above for several alternatives. To made a comparison we also simulate the power of the test of independence based on the Copula based measures sample Spearman's $\hat{\rho}$ and sample kendall's $t$ [ equation (27) and (28)] from the same sample and give them in a table along with the proposed tests. For testing these tests we use the simulated quantiles ( cut-offs ) of the distribution of the sample Spearman's $\hat{\rho}$ and the sample kendall's $t$ given in table $\ref{tab:cvr}$ and $\ref{tab:cvt}$ respectively. 
	
		\begin{table}[h]
	\centering
	\caption{The simulated quantiles of the sample distribution of the sample Spearman's $\hat{\rho}$ under null $H_0 : r=0$ for different sample size $n$ }
		\begin{tabular}{|c| c c c c c c c c c|} \hline
		n    &  2.5\% & 5\% & 10\% &  25\%  &  50\% & 75\% &  90\% & 95\% & 97.5\% \\ \hline
	 30 &  -0.3615 &  -0.3068 &  -0.2400 &  -0.1270 &   0.0007 &   0.1284 &   0.2396 &   0.3050 &   0.3597\\
   50 &  -0.2808 &  -0.2352 &  -0.1849 &  -0.0976 &   0.0011 &   0.0973 &   0.1839 &   0.2345 &   0.2792\\
   70 &  -0.2358 &  -0.1988 &  -0.1556 &  -0.0820 &  -0.0001 &   0.0819 &   0.1547 &   0.1981 &   0.2349\\
  100 &  -0.1976 &  -0.1656 &  -0.1291 &  -0.0682 &   0.0002 &   0.0684 &   0.1297 &   0.1660 &   0.1984\\
  150 &  -0.1602 &  -0.1347 &  -0.1049 &  -0.0555 &   0.0003 &   0.0555 &   0.1050 &   0.1350 &   0.1617\\
  200 &  -0.1389 &  -0.1169 &  -0.0913 &  -0.0486 &  -0.0001 &   0.0476 &   0.0908 &   0.1162 &   0.1380\\  \hline
		\end{tabular}
	\label{tab:cvr}
\end{table}

		\begin{table}[h]
	\centering
	\caption{The simulated quantiles of the sample distribution of the sample kendall's $t$ under null $H_0 : r=0$ for different sample size $n$ }
		\begin{tabular}{|c| c c c c c c c c c|} \hline
		n    &  2.5\% & 5\% & 10\% &  25\%  &  50\% & 75\% &  90\% & 95\% & 97.5\% \\ \hline
   30 &   -0.2506 &  -0.2138 &  -0.1632 &  -0.0851 &   0.0023 &   0.0897 &   0.1632 &   0.2138 &   0.2506\\
   50 &   -0.1918 &  -0.1608 &  -0.1265 &  -0.0661 &   0.0008 &   0.0661 &   0.1249 &   0.1608 &   0.1902\\
   70 &   -0.1602 &  -0.1346 &  -0.1056 &  -0.0551 &  -0.0004 &   0.0551 &   0.1048 &   0.1346 &   0.1594\\
  100 &   -0.1333 &  -0.1115 &  -0.0869 &  -0.0461 &        0 &   0.0461 &   0.0877 &   0.1119 &   0.1337\\
  150 &   -0.1078 &  -0.0906 &  -0.0706 &  -0.0373 &   0.0003 &   0.0371 &   0.0704 &   0.0906 &   0.1087\\
  200 &   -0.0934 &  -0.0784 &  -0.0612 &  -0.0325 &        0 &   0.0320 &   0.0608 &   0.0779 &   0.0929\\  \hline
		\end{tabular}
	\label{tab:cvt}
\end{table}

	Further we also simulate the power of the usual parametric test for $H_0:r=0$ against $H_1:r \neq 0$ based on the pearson correlation coefficient $r$ which is given by the test statistics :
	
\begin{equation}
	T_p = \frac{\sqrt{n-2}r}{\sqrt{1 - r^2}}
\end{equation}
where $n$ is the sample size. We know that the null distribution of $T_p$ is the student's $t$-distribution with degrees of freedoms $n-2$ and we use this distribution to get the cut-off for testing using $T_p$. In the following subsections we list down our findings about the power of the different tests for some simple class of alternatives.

\subsection{Power against the Correlated Normal alternatives}	
\paragraph{}
	Firstly we choose the simplest possible alternatives that the sample comes from a bivariate normal population with means $0$ and variances $1$ but with correlation coefficient $r$. We take different values of the $r$ in $[ -1, 1]$ and simulate the power of all the tests explained above for different sample sizes $n=30, 50, 70, 100, 150, 200 $. We tabulate the simulated powers respectively in the table \ref{tab:p30} to \ref{tab:p200} given in the Appendix.
	
		From the tabulated powers simulated from $10000$ iterations we make the following findings on the comparative study of the tests of independence.

\begin{itemize}
	\item As expected, the parametric test based on $T_p$ gives the best power in all sample against the correlated normal samples because this parametric test is indeed based on the normal population and so its null distribution is exact under the assumption of normality.

	\item The test based on the sample Spearman's $\hat{\rho}$ and the sample kendall's $t$ also gives quite good results comparable with that from the parametric test and is better than our proposed tests. This is because that under the assumption of normality the sample Spearman's $\hat{\rho}$ and the sample kendall's $t$  behaves quite similar to that of the pearson's usual correlation coefficients. But we will see that as the samples deviates from the assumption of normality the power of these tests based on the parametric statistic $T_p$ and the sample Spearman's $\hat{\rho}$ and the sample kendall's $t$ will reduces.
	
	\item Next we note that among the nonparametric test proposed in the above section 6, the test based on the trace of $RPM$ that  is the statistics $T_1$ has no power against the alternative of negative association. This coincides with our previous discussion in the sub-section 6.1. However for the alternatives of positive correlations the test based on the statistic $T_1$ gives quite good power and this power also increases with the sample size. Finally under the perfect positive association that is for one correlation this test gives the power of one because in that case the distribution of the $T_1$ becomes degenerate at its highest possible value of $\frac{n(n+1)}{2}$.
	
	\item All the other proposed tests based on the statistics $T_2$ to $T_6$ gives reasonable good power in both sided alternatives. To compare among them we plot the power of these tests for different sample sizes.[ Figure $\ref{fig:p1}$ in Appendix]
	
	\item From the figures $\ref{fig:p1}$ , it is clear that in all the cases the power increases with the sample size.
	 
	 \item All the tests based on $T_2$ to $T_6$ are symmetric in both sided alternatives.
	 
	 \item The test based on $T_5$ ( the entropy of the MFCCs of the $RPV$ ) gives the lowest power among these tests.
	 
	 \item The tests based on $T_3$, $T_4$ and $T_6$ are almost uniformly comparable with the ordering $T_4 \geq T_3 \geq T_6$ except a little fluctuation for the large magnitude of the correlation that may be just due to sampling fluctuation.
	 
	 \item The test based on $T_2$ gives the best power among these tests by $T_2$ to $T_6$ for very close alternatives with the correlation $r$ being in the region $[-0.6,0.6]$. Beyond this region the power of the tests $T_4$ and $T_3$ exceed that of the $T_2$ for small samples but for large sample sizes $T_2$ again gives the highest power. So we can infer that the test based on $T_2$ is locally best among all the proposed tests for the correlated normal alternative with correlation being in the rage of $[-0.6,0.6]$.
	 
	 \item But there is a problem with the test based on $T_2$. It gives $0$ power against the alternative of perfect association ( that is, $r=1$ and $r=-1$ ). However this don not gives rise to a serious barrier because we can directly separate out the case with the perfect association by only looking at the $RPV$ as explained in the sub-section 5.2.
	 
\end{itemize}

\subsection{Power against the Correlated Random-Walk Type Normal alternatives}	
\paragraph{}
	Now we slightly deviates from the normal population. Here we consider the correlated random-walk type normal sample $(X_1,Y_1),\cdots,(X_n,Y_n)$. That is, here $(X_i,Y_i)$'s are not directly a bivariate normal random variables. Rather they are of the form 
	$$
	(X_i,Y_i) = \sum_{j=1}^i~(Z_j,W_j)~~~~~~~~~~\forall ~~i
	$$
	where $(Z_j,W_j)$'s are the correlated bivariate normal random variables with means $0$, variances $1$ and correlation $r$. Note that under this set-up $(X_i,Y_i)$'s are not independent even for $r=0$. We now take different values of the $r$ in $[ -1, 1]$ and simulate the power of all the tests explained above for different sample sizes $n=30, 50, 70, 100, 150, 200 $. We tabulate the simulated powers respectively in the table \ref{tab:p130} to \ref{tab:p1200} in Appendix.
	
\paragraph{}
	From the tabulated powers simulated from $10000$ iterations we make the following findings on the comparative study of the tests of independence. Note that here the correlation $r=0$ does not implies independence as the samples are not normally distributed. However this situation is very close to the null hypothesis and so we emphasis the comparison of the powers in this particular local alternative.

\begin{itemize}
	\item Here also the parametric test based on $T_p$ gives the best power in all sample as again samples are related to the normal population so that the null distribution of $T_p$ is agin exact. However some of our proposed non-parametric test like that based on $T_4$ gives almost equal power as this parametric test.
	
	\item The test based on the sample Spearman's $\hat{\rho}$ and the sample kendall's $t$ also gives very good power that are almost equal to that from the parametric test. And the power obtained from all this tests converges to one for any value correlation as the sample size increases. Even for correlation $r=0$ this converges is very first.
	
	\item In this case also the test based on the statistic $T_1$ ( the trace of $RPM$ ) has no power against the alternative of negative association. However for the alternatives of positive correlations the test based on the statistic $T_1$ gives quite good power and this power also increases with the sample size. And similar to the previous case, under the perfect positive association this test gives the power of one because in that case the distribution of the $T_1$ becomes degenerate at its highest possible value of $\frac{n(n+1)}{2}$.
	
	\item All the other proposed tests based on the statistics $T_2$ to $T_6$ gives similar power in both sided alternatives. To compare among them we plot the power of these tests for different sample sizes.[ Figure $\ref{fig:p2}$ in Appendix ]
	
	\item All the tests based on $T_2$ to $T_6$ are symmetric in both sided alternatives.
	
	\item From the figures $\ref{fig:p1}$ , it is clear that in all the cases the power increases with the sample size.
	 
	 \item The test based on $T_5$ ( the entropy of the MFCCs of the $RPV$ ) has very less power. So we cannot use this statistic for testing independence against such kind of alternative.
	 
	 \item The tests based on $T_3$, $T_4$ and $T_6$ are uniformly comparable with the ordering $T_4 \geq T_3 \geq T_6$. However the difference between the powers obtained from this tests are seem to be independent of the sample size and tends to zero as the magnitude of the association tends to one because then all the powers tends to one.
	 
	 \item The test based on $T_2$ gives the very less power ( even less than that obtained from $T_5$ ) when the value of $r$ is near to zero. For small and moderate sample sizes, say for $n=100$ or less, it is less than the power obtained from all the three tests based on $T_3$, $T_4$ and $T_6$ for all possible values of $r$. And for large sample sizes also ( for $n=150,~200$ or more) the situation is same for $r$ in $[-0.6,0.6]$ and for $r$ beyond this interval the power of $T_2$ exceeds that of $T_4$. But again under perfect association ( $r=1$ and $r=-1$ ) the power of this test based on $T_2$ becomes zero.
	 
	 \item Thus among all the proposed non-parametric tests the test based on $T_4$ ( Sum of squares of the Fourier coefficients  of the $RPV$ ) has the maximum power except for the large sample size with $r$ outside the interval $[-0.6,0.6]$ where $T_2$ has more power. However the difference in the powers of $T_4$ and $T_2$ is very small in the above case and otherwise the power of $T_2$ is very less compared to the $T_4$. Therefore we can conclude that it is best to use $T_4$ among all the proposed test for testing independence against such random-walk type of alternatives.
	 
\end{itemize}

\newpage
\section{Discussion}
\paragraph{}
	In our work we give some method for estimating the copula based on a bivariate sample. When the full data is available we can easily estimate the copula by the Empirical Copula. In case the full data is not available we give another method of estimating the copula based on only some sufficient information about the data. Next we develop a general framework that covers all the permutation invariant rank-based statistics. For we define the Rank-position vector $RPV$ for a bivariate sample and prove that all the permutation invariant rank-based statistics are some functional of the $RPV$. This class of all rank-based statistics which are the functions of the $RPV$ also includes the well known Copula based estimators of the measure of association like Spearman's $\rho$, Kendall's $\tau$ and Gini's $\gamma$. We also find out the null distribution of the $RPV$ under the assumption of independence and see that this null distribution is independent of the underlying distribution of the bivariate random variables. Hence all the estimators or test statistics based on the $RPV$ will be non-parametric in nature. Next we propose some tests for independence based on $RPV$ from our intuitive justification and compare their relative performances by a simulation study. In this report we present the simulated critical values of all the proposed tests of independence for some sample sizes and also tabulate the power of these tests against two kind of alternative of non-independence --- one is the correlated normal alternative and second one is the correlated random-walk type normal alternatives. We see that the proposed test statistics are locally best for the some special kind of alternatives of non-independence. For example, $T_2$ is locally best for close alternatives of correlated bivariate normal samples, whereas $T_4$ is best in the second kind of alternative considered above.
	
		This work can be extended in various perspectives in the further study. Note that the class of all possible alternatives against independence is a very broad class and we here consider two types of alternatives from that class. So it will be a good idea to compare the performances of the proposed test statistics for various other kind of alternatives also in the future study. We can also try to find out the non-null distribution of the $RPV$. Also we may try to derive the null or non-null distributions of the proposed test statistics at least asymptotically.  One may also try develop the rank-based non-parametric test-statistics for testing independence from the class of all the functional of the $RPV$ that will be uniformly based against all possible alternatives of non-independence in that class. If such a test does not exists we may look for a locally best test against some specified alternative.
		
\vspace{1cm} 
		
\section{Acknowledgement}
\paragraph{}
	We are grateful to Prof. Debapriya Sengupta of Applied Statistics Unit (ASU) in the Indian Statistical Institute, Kolkata for his great guidance and help to complete this project.

\newpage

\section{Appendix  : Tables and Plots of the simulated Power of the different tests} 

\begin{table}[ht]
	\centering
	\caption{The simulated Powers of the Test of independence based on several test statistics against the alternative that the sample comes from bivariate normal with different correlation coefficients $r$ and for sample size $n=30$ }
		\begin{tabular}{|c| c c c c c c c c c|} \hline
		r    &  $T_1$ & $T_2$ & $T_3$ &  $T_4$  &  $T_5$ & $T_6$ &  $Spearman's~~ \hat{\rho}$ & $Kendall's~~ t$ & $T_p$ \\ \hline
	 -1.0000 &      0  &       0 &  1.0000 &  1.0000 &       0 &  1.0000 &  1.0000 &  1.0000 &  1.0000\\
   -0.9000 &  0.0159 &  0.4862 &  0.4744 &  0.5671 &  0.0775 &  0.4365 &  1.0000 &  1.0000 &  1.0000\\
   -0.8000 &  0.0183 &  0.3672 &  0.2058 &  0.2453 &  0.0721 &  0.1933 &  0.9996 &  0.9994 &  0.9998\\
   -0.7000 &  0.0203 &  0.2498 &  0.1110 &  0.1226 &  0.0667 &  0.1057 &  0.9881 &  0.9881 &  0.9958\\
   -0.6000 &  0.0251 &  0.1739 &  0.0713 &  0.0747 &  0.0565 &  0.0691 &  0.9196 &  0.9166 &  0.9551\\
   -0.5000 &  0.0320 &  0.1259 &  0.0605 &  0.0609 &  0.0487 &  0.0581 &  0.7792 &  0.7713 &  0.8359\\
   -0.4000 &  0.0359 &  0.0963 &  0.0493 &  0.0503 &  0.0504 &  0.0507 &  0.5484 &  0.5406 &  0.6129\\
   -0.3000 &  0.0431 &  0.0701 &  0.0503 &  0.0501 &  0.0504 &  0.0496 &  0.3229 &  0.3133 &  0.3605\\
   -0.2000 &  0.0503 &  0.0592 &  0.0528 &  0.0530 &  0.0504 &  0.0528 &  0.1659 &  0.1610 &  0.1831\\
   -0.1000 &  0.0658 &  0.0545 &  0.0488 &  0.0479 &  0.0479 &  0.0496 &  0.0781 &  0.0759 &  0.0794\\
        0  &  0.0802 &  0.0486 &  0.0516 &  0.0496 &  0.0479 &  0.0514 &  0.0553 &  0.0555 &  0.0534\\
    0.1000 &  0.0957 &  0.0527 &  0.0481 &  0.0498 &  0.0506 &  0.0485 &  0.0789 &  0.0760 &  0.0780\\
    0.2000 &  0.1296 &  0.0615 &  0.0462 &  0.0470 &  0.0487 &  0.0471 &  0.1739 &  0.1690 &  0.1893\\
    0.3000 &  0.1616 &  0.0737 &  0.0503 &  0.0511 &  0.0532 &  0.0511 &  0.3436 &  0.3335 &  0.3811\\
    0.4000 &  0.2045 &  0.0937 &  0.0500 &  0.0495 &  0.0550 &  0.0504 &  0.5542 &  0.5432 &  0.6070\\
    0.5000 &  0.2578 &  0.1210 &  0.0569 &  0.0586 &  0.0543 &  0.0580 &  0.7728 &  0.7643 &  0.8263\\
    0.6000 &  0.3489 &  0.1815 &  0.0757 &  0.0792 &  0.0626 &  0.0730 &  0.9283 &  0.9235 &  0.9570\\
    0.7000 &  0.4437 &  0.2478 &  0.1031 &  0.1158 &  0.0659 &  0.0985 &  0.9880 &  0.9869 &  0.9959\\
    0.8000 &  0.5965 &  0.3605 &  0.1980 &  0.2356 &  0.0697 &  0.1826 &  0.9995 &  0.9995 &  0.9998\\
    0.9000 &  0.8213 &  0.4792 &  0.4883 &  0.5754 &  0.0842 &  0.4525 &  1.0000 &  1.0000 &  1.0000\\
    1.0000 &  1.0000 &       0 &  1.0000 &  1.0000 &       0 &  1.0000 &  1.0000 &  1.0000 &  1.0000\\ \hline
		\end{tabular}
	\label{tab:p30}
\end{table}

\begin{table}[h]
	\centering
	\caption{The simulated Powers of the Test of independence based on several test statistics against the alternative that the sample comes from bivariate normal with different correlation coefficients $r$ and for sample size $n=50$ }
		\begin{tabular}{|c| c c c c c c c c c|} \hline
		r    &  $T_1$ & $T_2$ & $T_3$ &  $T_4$  &  $T_5$ & $T_6$ &  $Spearman's~~ \hat{\rho}$ & $Kendall's~~ t$ & $T_p$ \\ \hline
   -1.0000 &       0 &       0 &  1.0000 &  1.0000 &       0 &  1.0000 &  1.0000 &  1.0000 &  1.0000\\
   -0.9000 &  0.0148 &  0.6124 &  0.9170 &  0.9672 &  0.1211 &  0.8832 &  1.0000 &  1.0000 &  1.0000\\
   -0.8000 &  0.0201 &  0.5739 &  0.5617 &  0.6610 &  0.0987 &  0.5085 &  1.0000 &  1.0000 &  1.0000\\
   -0.7000 &  0.0235 &  0.4266 &  0.2742 &  0.3326 &  0.0795 &  0.2457 &  0.9995 &  0.9995 &  0.9999\\
   -0.6000 &  0.0255 &  0.2913 &  0.1476 &  0.1730 &  0.0664 &  0.1346 &  0.9933 &  0.9927 &  0.9976\\
   -0.5000 &  0.0310 &  0.1880 &  0.0899 &  0.0974 &  0.0543 &  0.0830 &  0.9472 &  0.9456 &  0.9702\\
   -0.4000 &  0.0375 &  0.1193 &  0.0616 &  0.0646 &  0.0567 &  0.0603 &  0.7757 &  0.7730 &  0.8291\\
   -0.3000 &  0.0439 &  0.0840 &  0.0520 &  0.0525 &  0.0508 &  0.0528 &  0.5243 &  0.5212 &  0.5803\\
   -0.2000 &  0.0566 &  0.0636 &  0.0483 &  0.0487 &  0.0453 &  0.0493 &  0.2612 &  0.2553 &  0.2884\\
   -0.1000 &  0.0637 &  0.0586 &  0.0529 &  0.0530 &  0.0494 &  0.0535 &  0.1026 &  0.1007 &  0.1133\\
         0 &  0.0869 &  0.0483 &  0.0473 &  0.0471 &  0.0450 &  0.0465 &  0.0568 &  0.0554 &  0.0548\\
    0.1000 &  0.0980 &  0.0513 &  0.0515 &  0.0496 &  0.0460 &  0.0507 &  0.1002 &  0.0992 &  0.1083\\
    0.2000 &  0.1285 &  0.0594 &  0.0504 &  0.0521 &  0.0477 &  0.0485 &  0.2551 &  0.2523 &  0.2807\\
    0.3000 &  0.1616 &  0.0855 &  0.0546 &  0.0563 &  0.0471 &  0.0543 &  0.5160 &  0.5140 &  0.5722\\
    0.4000 &  0.2075 &  0.1213 &  0.0619 &  0.0663 &  0.0504 &  0.0623 &  0.7863 &  0.7843 &  0.8321\\
    0.5000 &  0.2779 &  0.1863 &  0.0889 &  0.0970 &  0.0581 &  0.0829 &  0.9461 &  0.9470 &  0.9660\\
    0.6000 &  0.3676 &  0.2886 &  0.1377 &  0.1619 &  0.0644 &  0.1259 &  0.9945 &  0.9941 &  0.9972\\
    0.7000 &  0.4772 &  0.4290 &  0.2781 &  0.3451 &  0.0760 &  0.2490 &  0.9996 &  0.9998 &  1.0000\\
    0.8000 &  0.6385 &  0.5746 &  0.5581 &  0.6609 &  0.1045 &  0.5006 &  1.0000 &  1.0000 &  1.0000\\
    0.9000 &  0.8511 &  0.6127 &  0.9214 &  0.9673 &  0.1184 &  0.8874 &  1.0000 &  1.0000 &  1.0000\\
    1.0000 &  1.0000 &       0 &  1.0000 &  1.0000 &       0 &  1.0000 &  1.0000 &  1.0000 &  1.0000\\ \hline
		\end{tabular}
	\label{tab:p50}
\end{table}

\begin{table}[h]
	\centering
	\caption{The simulated Powers of the Test of independence based on several test statistics against the alternative that the sample comes from bivariate normal with different correlation coefficients $r$ and for sample size $n=70$ }
		\begin{tabular}{|c| c c c c c c c c c|} \hline
		r    &  $T_1$ & $T_2$ & $T_3$ &  $T_4$  &  $T_5$ & $T_6$ &  $Spearman's~~ \hat{\rho}$ & $Kendall's~~ t$ & $T_p$ \\ \hline
   -1.0000 &       0 &       0 &  1.0000 &  1.0000 &       0 &  1.0000 &  1.0000 &  1.0000 &  1.0000\\
   -0.9000 &  0.0167 &  0.7330 &  0.9964 &  0.9992 &  0.1580 &  0.9916 &  1.0000 &  1.0000 &  1.0000\\
   -0.8000 &  0.0174 &  0.7736 &  0.8283 &  0.9076 &  0.1172 &  0.7700 &  1.0000 &  1.0000 &  1.0000\\
   -0.7000 &  0.0202 &  0.6453 &  0.4707 &  0.5835 &  0.0873 &  0.4146 &  1.0000 &  1.0000 &  1.0000\\
   -0.6000 &  0.0263 &  0.4385 &  0.2385 &  0.2956 &  0.0736 &  0.2066 &  0.9994 &  0.9994 &  0.9998\\
   -0.5000 &  0.0307 &  0.2756 &  0.1275 &  0.1480 &  0.0623 &  0.1178 &  0.9887 &  0.9882 &  0.9947\\
   -0.4000 &  0.0337 &  0.1646 &  0.0791 &  0.0889 &  0.0569 &  0.0740 &  0.9121 &  0.9109 &  0.9406\\
   -0.3000 &  0.0429 &  0.1067 &  0.0551 &  0.0573 &  0.0460 &  0.0536 &  0.6779 &  0.6748 &  0.7262\\
   -0.2000 &  0.0537 &  0.0754 &  0.0512 &  0.0528 &  0.0525 &  0.0500 &  0.3515 &  0.3483 &  0.3874\\
   -0.1000 &  0.0640 &  0.0555 &  0.0485 &  0.0496 &  0.0455 &  0.0488 &  0.1204 &  0.1180 &  0.1280\\
         0 &  0.0802 &  0.0489 &  0.0518 &  0.0498 &  0.0503 &  0.0507 &  0.0453 &  0.0449 &  0.0488\\
    0.1000 &  0.0988 &  0.0550 &  0.0520 &  0.0512 &  0.0553 &  0.0496 &  0.1214 &  0.1190 &  0.1303\\
    0.2000 &  0.1326 &  0.0632 &  0.0508 &  0.0499 &  0.0477 &  0.0491 &  0.3609 &  0.3603 &  0.3925\\
    0.3000 &  0.1693 &  0.1042 &  0.0545 &  0.0543 &  0.0507 &  0.0527 &  0.6727 &  0.6720 &  0.7205\\
    0.4000 &  0.2109 &  0.1621 &  0.0776 &  0.0839 &  0.0585 &  0.0757 &  0.9079 &  0.9056 &  0.9357\\
    0.5000 &  0.2759 &  0.2748 &  0.1260 &  0.1499 &  0.0649 &  0.1137 &  0.9883 &  0.9891 &  0.9949\\
    0.6000 &  0.3704 &  0.4367 &  0.2442 &  0.3009 &  0.0728 &  0.2131 &  0.9996 &  0.9997 &  1.0000\\
    0.7000 &  0.4819 &  0.6404 &  0.4656 &  0.5817 &  0.0929 &  0.4059 &  1.0000 &  1.0000 &  1.0000\\
    0.8000 &  0.6500 &  0.7803 &  0.8275 &  0.9097 &  0.1248 &  0.7661 &  1.0000 &  1.0000 &  1.0000\\
    0.9000 &  0.8676 &  0.7331 &  0.9966 &  0.9996 &  0.1644 &  0.9909 &  1.0000 &  1.0000 &  1.0000\\
    1.0000 &  1.0000 &       0 &  1.0000 &  1.0000 &       0 &  1.0000 &  1.0000 &  1.0000 &  1.0000\\ \hline
		\end{tabular}
	\label{tab:p70}
\end{table}

\begin{table}[h]
	\centering
	\caption{The simulated Powers of the Test of independence based on several test statistics against the alternative that the sample comes from bivariate normal with different correlation coefficients $r$ and for sample size $n=100$ }
		\begin{tabular}{|c| c c c c c c c c c|} \hline
		r    &  $T_1$ & $T_2$ & $T_3$ &  $T_4$  &  $T_5$ & $T_6$ &  $Spearman's~~ \hat{\rho}$ & $Kendall's~~ t$ & $T_p$ \\ \hline
   -1.0000 &       0 &       0 &  1.0000 &  1.0000 &       0 &  1.0000 &  1.0000 &  1.0000 &  1.0000\\
   -0.9000 &  0.0174 &  0.8593 &  1.0000 &  1.0000 &  0.2203 &  0.9998 &  1.0000 &  1.0000 &  1.0000\\
   -0.8000 &  0.0189 &  0.9134 &  0.9752 &  0.9936 &  0.1520 &  0.9493 &  1.0000 &  1.0000 &  1.0000\\
   -0.7000 &  0.0232 &  0.8411 &  0.7394 &  0.8522 &  0.1153 &  0.6617 &  1.0000 &  1.0000 &  1.0000\\
   -0.6000 &  0.0261 &  0.6366 &  0.4061 &  0.5153 &  0.0871 &  0.3460 &  1.0000 &  1.0000 &  1.0000\\
   -0.5000 &  0.0300 &  0.4081 &  0.1977 &  0.2490 &  0.0719 &  0.1664 &  0.9994 &  0.9995 &  0.9996\\
   -0.4000 &  0.0345 &  0.2347 &  0.1027 &  0.1163 &  0.0531 &  0.0930 &  0.9741 &  0.9733 &  0.9846\\
   -0.3000 &  0.0440 &  0.1304 &  0.0708 &  0.0727 &  0.0562 &  0.0671 &  0.8252 &  0.8250 &  0.8648\\
   -0.2000 &  0.0528 &  0.0780 &  0.0536 &  0.0539 &  0.0488 &  0.0528 &  0.4802 &  0.4790 &  0.5265\\
   -0.1000 &  0.0664 &  0.0493 &  0.0473 &  0.0465 &  0.0472 &  0.0454 &  0.1508 &  0.1509 &  0.1661\\
         0 &  0.0754 &  0.0492 &  0.0498 &  0.0497 &  0.0506 &  0.0495 &  0.0500 &  0.0499 &  0.0493\\
    0.1000 &  0.1008 &  0.0538 &  0.0519 &  0.0508 &  0.0524 &  0.0511 &  0.1473 &  0.1468 &  0.1593\\
    0.2000 &  0.1256 &  0.0796 &  0.0532 &  0.0541 &  0.0513 &  0.0518 &  0.4705 &  0.4709 &  0.5181\\
    0.3000 &  0.1695 &  0.1348 &  0.0683 &  0.0736 &  0.0525 &  0.0669 &  0.8168 &  0.8154 &  0.8606\\
    0.4000 &  0.2162 &  0.2274 &  0.1003 &  0.1171 &  0.0571 &  0.0895 &  0.9757 &  0.9764 &  0.9866\\
    0.5000 &  0.2905 &  0.4063 &  0.1938 &  0.2471 &  0.0679 &  0.1683 &  0.9993 &  0.9991 &  0.9997\\
    0.6000 &  0.3731 &  0.6246 &  0.4025 &  0.5129 &  0.0850 &  0.3395 &  1.0000 &  1.0000 &  1.0000\\
    0.7000 &  0.5020 &  0.8365 &  0.7366 &  0.8523 &  0.1083 &  0.6573 &  1.0000 &  1.0000 &  1.0000\\
    0.8000 &  0.6689 &  0.9090 &  0.9730 &  0.9926 &  0.1537 &  0.9447 &  1.0000 &  1.0000 &  1.0000\\
    0.9000 &  0.8843 &  0.8634 &  1.0000 &  1.0000 &  0.2263 &  1.0000 &  1.0000 &  1.0000 &  1.0000\\
    1.0000 &  1.0000 &       0 &  1.0000 &  1.0000 &       0 &  1.0000 &  1.0000 &  1.0000 &  1.0000\\ \hline
		\end{tabular}
	\label{tab:p100}
\end{table}

\begin{table}[h]
	\centering
	\caption{The simulated Powers of the Test of independence based on several test statistics against the alternative that the sample comes from bivariate normal with different correlation coefficients $r$ and for sample size $n=150$ }
		\begin{tabular}{|c| c c c c c c c c c|} \hline
		r    &  $T_1$ & $T_2$ & $T_3$ &  $T_4$  &  $T_5$ & $T_6$ &  $Spearman's~~ \hat{\rho}$ & $Kendall's~~ t$ & $T_p$ \\ \hline
   -1.0000 &       0 &       0 &  1.0000 &  1.0000 &       0 &  1.0000 &  1.0000 &  1.0000 &  1.0000\\
   -0.9000 &  0.0164 &  0.9476 &  1.0000 &  1.0000 &  0.3146 &  1.0000 &  1.0000 &  1.0000 &  1.0000\\
   -0.8000 &  0.0197 &  0.9778 &  0.9997 &  1.0000 &  0.1991 &  0.9980 &  1.0000 &  1.0000 &  1.0000\\
   -0.7000 &  0.0224 &  0.9660 &  0.9455 &  0.9881 &  0.1367 &  0.8919 &  1.0000 &  1.0000 &  1.0000\\
   -0.6000 &  0.0261 &  0.8606 &  0.6567 &  0.7988 &  0.0908 &  0.5599 &  1.0000 &  1.0000 &  1.0000\\
   -0.5000 &  0.0294 &  0.6236 &  0.3383 &  0.4572 &  0.0747 &  0.2739 &  1.0000 &  1.0000 &  1.0000\\
   -0.4000 &  0.0349 &  0.3587 &  0.1508 &  0.1907 &  0.0663 &  0.1277 &  0.9976 &  0.9975 &  0.9991\\
   -0.3000 &  0.0454 &  0.1852 &  0.0771 &  0.0878 &  0.0521 &  0.0723 &  0.9462 &  0.9459 &  0.9627\\
   -0.2000 &  0.0553 &  0.1032 &  0.0603 &  0.0579 &  0.0514 &  0.0561 &  0.6492 &  0.6507 &  0.6915\\
   -0.1000 &  0.0648 &  0.0596 &  0.0474 &  0.0450 &  0.0491 &  0.0483 &  0.2097 &  0.2080 &  0.2302\\
         0 &  0.0842 &  0.0511 &  0.0520 &  0.0523 &  0.0464 &  0.0508 &  0.0481 &  0.0477 &  0.0478\\
    0.1000 &  0.1027 &  0.0553 &  0.0516 &  0.0491 &  0.0510 &  0.0510 &  0.2035 &  0.2044 &  0.2229\\
    0.2000 &  0.1344 &  0.0987 &  0.0551 &  0.0559 &  0.0505 &  0.0544 &  0.6433 &  0.6421 &  0.6945\\
    0.3000 &  0.1628 &  0.1857 &  0.0766 &  0.0846 &  0.0562 &  0.0721 &  0.9446 &  0.9453 &  0.9646\\
    0.4000 &  0.2224 &  0.3546 &  0.1450 &  0.1867 &  0.0616 &  0.1246 &  0.9981 &  0.9980 &  0.9996\\
    0.5000 &  0.2939 &  0.6191 &  0.3304 &  0.4488 &  0.0736 &  0.2704 &  1.0000 &  1.0000 &  1.0000\\
    0.6000 &  0.3902 &  0.8616 &  0.6545 &  0.7995 &  0.0939 &  0.5595 &  1.0000 &  1.0000 &  1.0000\\
    0.7000 &  0.5277 &  0.9641 &  0.9400 &  0.9854 &  0.1309 &  0.8905 &  1.0000 &  1.0000 &  1.0000\\
    0.8000 &  0.6912 &  0.9772 &  0.9997 &  1.0000 &  0.1938 &  0.9990 &  1.0000 &  1.0000 &  1.0000\\
    0.9000 &  0.8957 &  0.9485 &  1.0000 &  1.0000 &  0.3169 &  1.0000 &  1.0000 &  1.0000 &  1.0000\\
    1.0000 &  1.0000 &       0 &  1.0000 &  1.0000 &       0 &  1.0000 &  1.0000 &  1.0000 &  1.0000\\ \hline
		\end{tabular}
	\label{tab:p150}
\end{table}

\begin{table}[h]
	\centering
	\caption{The simulated Powers of the Test of independence based on several test statistics against the alternative that the sample comes from bivariate normal with different correlation coefficients $r$ and for sample size $n=200$ }
		\begin{tabular}{|c| c c c c c c c c c|} \hline
		r    &  $T_1$ & $T_2$ & $T_3$ &  $T_4$  &  $T_5$ & $T_6$ &  $Spearman's~~ \hat{\rho}$ & $Kendall's~~ t$ & $T_p$ \\ \hline
   -1.0000 &       0 &       0 &  1.0000 &  1.0000 &       0 &  1.0000 &  1.0000 &  1.0000 &  1.0000\\
   -0.9000 &  0.0179 &  0.9812 &  1.0000 &  1.0000 &  0.4119 &  1.0000 &  1.0000 &  1.0000 &  1.0000\\
   -0.8000 &  0.0194 &  0.9937 &  1.0000 &  1.0000 &  0.2366 &  0.9999 &  1.0000 &  1.0000 &  1.0000\\
   -0.7000 &  0.0199 &  0.9940 &  0.9909 &  0.9988 &  0.1548 &  0.9725 &  1.0000 &  1.0000 &  1.0000\\
   -0.6000 &  0.0237 &  0.9525 &  0.8365 &  0.9411 &  0.1082 &  0.7309 &  1.0000 &  1.0000 &  1.0000\\
   -0.5000 &  0.0291 &  0.7801 &  0.4862 &  0.6377 &  0.0843 &  0.3907 &  1.0000 &  1.0000 &  1.0000\\
   -0.4000 &  0.0325 &  0.4775 &  0.2077 &  0.2785 &  0.0666 &  0.1662 &  1.0000 &  1.0000 &  1.0000\\
   -0.3000 &  0.0415 &  0.2495 &  0.1025 &  0.1223 &  0.0589 &  0.0911 &  0.9842 &  0.9844 &  0.9917\\
   -0.2000 &  0.0517 &  0.1176 &  0.0615 &  0.0612 &  0.0533 &  0.0592 &  0.7719 &  0.7714 &  0.8147\\
   -0.1000 &  0.0632 &  0.0681 &  0.0493 &  0.0511 &  0.0482 &  0.0479 &  0.2764 &  0.2740 &  0.2971\\
         0 &  0.0800 &  0.0501 &  0.0545 &  0.0522 &  0.0533 &  0.0530 &  0.0520 &  0.0514 &  0.0523\\
    0.1000 &  0.0979 &  0.0648 &  0.0494 &  0.0482 &  0.0509 &  0.0499 &  0.2694 &  0.2682 &  0.2884\\
    0.2000 &  0.1270 &  0.1163 &  0.0592 &  0.0633 &  0.0535 &  0.0578 &  0.7762 &  0.7739 &  0.8133\\
    0.3000 &  0.1682 &  0.2380 &  0.0947 &  0.1113 &  0.0589 &  0.0859 &  0.9850 &  0.9851 &  0.9923\\
    0.4000 &  0.2148 &  0.4785 &  0.2139 &  0.2871 &  0.0699 &  0.1740 &  1.0000 &  1.0000 &  1.0000\\
    0.5000 &  0.2948 &  0.7727 &  0.4723 &  0.6308 &  0.0824 &  0.3747 &  1.0000 &  1.0000 &  1.0000\\
    0.6000 &  0.3906 &  0.9502 &  0.8328 &  0.9415 &  0.1066 &  0.7336 &  1.0000 &  1.0000 &  1.0000\\
    0.7000 &  0.5393 &  0.9936 &  0.9922 &  0.9993 &  0.1562 &  0.9719 &  1.0000 &  1.0000 &  1.0000\\
    0.8000 &  0.7003 &  0.9957 &  1.0000 &  1.0000 &  0.2427 &  1.0000 &  1.0000 &  1.0000 &  1.0000\\
    0.9000 &  0.9033 &  0.9804 &  1.0000 &  1.0000 &  0.4118 &  1.0000 &  1.0000 &  1.0000 &  1.0000\\
    1.0000 &  1.0000 &       0 &  1.0000 &  1.0000 &       0 &  1.0000 &  1.0000 &  1.0000 &  1.0000\\\hline
		\end{tabular}
	\label{tab:p200}
\end{table}


\begin{table}[h]
	\centering
	\caption{The simulated Powers of the Test of independence based on several test statistics against the correlated random-walk type normal alternative with different correlation coefficients $r$ and for sample size $n=30$ }
		\begin{tabular}{|c| c c c c c c c c c|} \hline
		r    &  $T_1$ & $T_2$ & $T_3$ &  $T_4$  &  $T_5$ & $T_6$ &  $Spearman's~~ \hat{\rho}$ & $Kendall's~~ t$ & $T_p$ \\ \hline
	 -1.0000 &       0 &       0 &  1.0000 &  1.0000 &       0 &  1.0000 &  1.0000 &  1.0000 &  1.0000\\
   -0.9000 &  0.0141 &  0.4349 &  0.5318 &  0.5952 &  0.0620 &  0.5071 &  0.9775 &  0.9787 &  0.9807\\
   -0.8000 &  0.0220 &  0.3441 &  0.3546 &  0.4045 &  0.0671 &  0.3335 &  0.9040 &  0.9084 &  0.9159\\
   -0.7000 &  0.0289 &  0.2542 &  0.2710 &  0.3116 &  0.0682 &  0.2579 &  0.8359 &  0.8386 &  0.8492\\
   -0.6000 &  0.0341 &  0.1767 &  0.2148 &  0.2539 &  0.0612 &  0.2013 &  0.7626 &  0.7590 &  0.7749\\
   -0.5000 &  0.0402 &  0.1307 &  0.1793 &  0.2148 &  0.0642 &  0.1688 &  0.6883 &  0.6845 &  0.7129\\
   -0.4000 &  0.0535 &  0.1026 &  0.1578 &  0.1833 &  0.0641 &  0.1484 &  0.6326 &  0.6237 &  0.6581\\
   -0.3000 &  0.0626 &  0.0817 &  0.1419 &  0.1676 &  0.0661 &  0.1327 &  0.5962 &  0.5856 &  0.6207\\
   -0.2000 &  0.0819 &  0.0710 &  0.1312 &  0.1527 &  0.0626 &  0.1218 &  0.5692 &  0.5519 &  0.5927\\
   -0.1000 &  0.1079 &  0.0560 &  0.1266 &  0.1486 &  0.0571 &  0.1206 &  0.5487 &  0.5313 &  0.5784\\
         0 &  0.1231 &  0.0603 &  0.1298 &  0.1464 &  0.0634 &  0.1226 &  0.5371 &  0.5195 &  0.5671\\
    0.1000 &  0.1623 &  0.0615 &  0.1277 &  0.1486 &  0.0620 &  0.1207 &  0.5479 &  0.5316 &  0.5750\\
    0.2000 &  0.1862 &  0.0664 &  0.1336 &  0.1553 &  0.0610 &  0.1256 &  0.5618 &  0.5428 &  0.5848\\
    0.3000 &  0.2268 &  0.0798 &  0.1465 &  0.1681 &  0.0646 &  0.1378 &  0.5908 &  0.5785 &  0.6158\\
    0.4000 &  0.2754 &  0.1073 &  0.1595 &  0.1893 &  0.0645 &  0.1482 &  0.6354 &  0.6271 &  0.6589\\
    0.5000 &  0.3357 &  0.1335 &  0.1773 &  0.2075 &  0.0602 &  0.1677 &  0.6942 &  0.6866 &  0.7124\\
    0.6000 &  0.4250 &  0.1817 &  0.2169 &  0.2530 &  0.0671 &  0.2045 &  0.7529 &  0.7515 &  0.7744\\
    0.7000 &  0.5076 &  0.2489 &  0.2650 &  0.3124 &  0.0639 &  0.2500 &  0.8345 &  0.8347 &  0.8459\\
    0.8000 &  0.6378 &  0.3342 &  0.3609 &  0.4165 &  0.0680 &  0.3397 &  0.9078 &  0.9105 &  0.9177\\
    0.9000 &  0.7985 &  0.4441 &  0.5333 &  0.5973 &  0.0598 &  0.5101 &  0.9793 &  0.9804 &  0.9823\\
    1.0000 &  1.0000 &       0 &  1.0000 &  1.0000 &       0 &  1.0000 &  1.0000 &  1.0000 &  1.0000\\	 \hline
		\end{tabular}
	\label{tab:p130}
\end{table}

\begin{table}[h]
	\centering
	\caption{The simulated Powers of the Test of independence based on several test statistics against the correlated random-walk type normal alternative with different correlation coefficients $r$ and for sample size $n=50$ }
		\begin{tabular}{|c| c c c c c c c c c|} \hline
		r    &  $T_1$ & $T_2$ & $T_3$ &  $T_4$  &  $T_5$ & $T_6$ &  $Spearman's~~ \hat{\rho}$ & $Kendall's~~ t$ & $T_p$ \\ \hline
   -1.0000 &       0 &       0 &  1.0000 &  1.0000 &       0 &  1.0000 &  1.0000 &  1.0000 &  1.0000\\
   -0.9000 &  0.0153 &  0.5838 &  0.8278 &  0.8577 &  0.0920 &  0.8053 &  0.9849 &  0.9878 &  0.9877\\
   -0.8000 &  0.0186 &  0.5116 &  0.6655 &  0.7057 &  0.0994 &  0.6395 &  0.9395 &  0.9440 &  0.9431\\
   -0.7000 &  0.0270 &  0.3920 &  0.5482 &  0.5908 &  0.0940 &  0.5240 &  0.8801 &  0.8827 &  0.8887\\
   -0.6000 &  0.0331 &  0.2736 &  0.4617 &  0.5062 &  0.0882 &  0.4356 &  0.8175 &  0.8174 &  0.8324\\
   -0.5000 &  0.0407 &  0.1891 &  0.4076 &  0.4529 &  0.0826 &  0.3829 &  0.7689 &  0.7661 &  0.7865\\
   -0.4000 &  0.0535 &  0.1314 &  0.3746 &  0.4135 &  0.0908 &  0.3522 &  0.7205 &  0.7129 &  0.7382\\
   -0.3000 &  0.0724 &  0.0916 &  0.3336 &  0.3748 &  0.0806 &  0.3137 &  0.6790 &  0.6705 &  0.6999\\
   -0.2000 &  0.0870 &  0.0747 &  0.3204 &  0.3595 &  0.0788 &  0.2981 &  0.6689 &  0.6566 &  0.6935\\
   -0.1000 &  0.1108 &  0.0565 &  0.3164 &  0.3562 &  0.0835 &  0.2969 &  0.6488 &  0.6382 &  0.6741\\
         0 &  0.1326 &  0.0561 &  0.2994 &  0.3377 &  0.0796 &  0.2820 &  0.6503 &  0.6347 &  0.6756\\
    0.1000 &  0.1553 &  0.0580 &  0.3083 &  0.3474 &  0.0795 &  0.2875 &  0.6432 &  0.6312 &  0.6686\\
    0.2000 &  0.1949 &  0.0725 &  0.3216 &  0.3648 &  0.0791 &  0.2991 &  0.6648 &  0.6537 &  0.6869\\
    0.3000 &  0.2369 &  0.0925 &  0.3399 &  0.3773 &  0.0828 &  0.3152 &  0.6808 &  0.6727 &  0.7047\\
    0.4000 &  0.2864 &  0.1302 &  0.3634 &  0.4024 &  0.0865 &  0.3406 &  0.7180 &  0.7131 &  0.7362\\
    0.5000 &  0.3450 &  0.1868 &  0.4074 &  0.4506 &  0.0836 &  0.3808 &  0.7618 &  0.7580 &  0.7791\\
    0.6000 &  0.4271 &  0.2691 &  0.4701 &  0.5136 &  0.0898 &  0.4470 &  0.8227 &  0.8233 &  0.8334\\
    0.7000 &  0.5342 &  0.3837 &  0.5495 &  0.5938 &  0.0873 &  0.5221 &  0.8796 &  0.8843 &  0.8868\\
    0.8000 &  0.6545 &  0.5026 &  0.6557 &  0.6989 &  0.0941 &  0.6292 &  0.9375 &  0.9419 &  0.9431\\
    0.9000 &  0.8239 &  0.5854 &  0.8285 &  0.8596 &  0.0970 &  0.8057 &  0.9868 &  0.9885 &  0.9883\\
    1.0000 &  1.0000 &       0 &  1.0000 &  1.0000 &       0 &  1.0000 &  1.0000 &  1.0000 &  1.0000\\\hline
		\end{tabular}
	\label{tab:p150}
\end{table}

\begin{table}[h]
	\centering
	\caption{The simulated Powers of the Test of independence based on several test statistics against the  correlated random-walk type normal alternative with different correlation coefficients $r$ and for sample size $n=70$ }
		\begin{tabular}{|c| c c c c c c c c c|} \hline
		r    &  $T_1$ & $T_2$ & $T_3$ &  $T_4$  &  $T_5$ & $T_6$ &  $Spearman's~~ \hat{\rho}$ & $Kendall's~~ t$ & $T_p$ \\ \hline
   -1.0000 &       0 &       0 &  1.0000 &  1.0000 &       0 &  1.0000 &  1.0000 &  1.0000 &  1.0000\\
   -0.9000 &  0.0172 &  0.7037 &  0.9050 &  0.9223 &  0.1282 &  0.8890 &  0.9901 &  0.9918 &  0.9910\\
   -0.8000 &  0.0204 &  0.6911 &  0.7752 &  0.8073 &  0.1238 &  0.7532 &  0.9511 &  0.9552 &  0.9540\\
   -0.7000 &  0.0242 &  0.5642 &  0.6752 &  0.7139 &  0.1149 &  0.6506 &  0.9011 &  0.9049 &  0.9076\\
   -0.6000 &  0.0332 &  0.3929 &  0.6017 &  0.6430 &  0.1084 &  0.5747 &  0.8501 &  0.8514 &  0.8613\\
   -0.5000 &  0.0390 &  0.2664 &  0.5363 &  0.5785 &  0.0996 &  0.5067 &  0.8021 &  0.8006 &  0.8142\\
   -0.4000 &  0.0554 &  0.1852 &  0.4951 &  0.5405 &  0.1012 &  0.4679 &  0.7603 &  0.7559 &  0.7724\\
   -0.3000 &  0.0671 &  0.1220 &  0.4740 &  0.5178 &  0.1013 &  0.4449 &  0.7347 &  0.7289 &  0.7538\\
   -0.2000 &  0.0843 &  0.0782 &  0.4425 &  0.4863 &  0.0987 &  0.4159 &  0.7115 &  0.7041 &  0.7302\\
   -0.1000 &  0.1032 &  0.0584 &  0.4236 &  0.4722 &  0.1007 &  0.3994 &  0.6904 &  0.6760 &  0.7102\\
         0 &  0.1294 &  0.0539 &  0.4404 &  0.4846 &  0.0990 &  0.4129 &  0.6971 &  0.6846 &  0.7167\\
    0.1000 &  0.1603 &  0.0612 &  0.4377 &  0.4844 &  0.0940 &  0.4139 &  0.7058 &  0.6898 &  0.7200\\
    0.2000 &  0.1889 &  0.0837 &  0.4455 &  0.4939 &  0.0965 &  0.4183 &  0.7178 &  0.7077 &  0.7344\\
    0.3000 &  0.2407 &  0.1213 &  0.4695 &  0.5166 &  0.0945 &  0.4413 &  0.7391 &  0.7293 &  0.7538\\
    0.4000 &  0.2909 &  0.1815 &  0.4943 &  0.5422 &  0.1056 &  0.4696 &  0.7648 &  0.7612 &  0.7787\\
    0.5000 &  0.3544 &  0.2663 &  0.5349 &  0.5793 &  0.1049 &  0.5089 &  0.8049 &  0.8028 &  0.8146\\
    0.6000 &  0.4342 &  0.4031 &  0.6007 &  0.6453 &  0.1118 &  0.5705 &  0.8447 &  0.8476 &  0.8533\\
    0.7000 &  0.5347 &  0.5539 &  0.6764 &  0.7154 &  0.1217 &  0.6507 &  0.8977 &  0.9024 &  0.9028\\
    0.8000 &  0.6732 &  0.6918 &  0.7794 &  0.8099 &  0.1236 &  0.7588 &  0.9515 &  0.9557 &  0.9549\\
    0.9000 &  0.8310 &  0.6892 &  0.9039 &  0.9200 &  0.1313 &  0.8890 &  0.9879 &  0.9902 &  0.9900\\
    1.0000 &  1.0000 &       0 &  1.0000 &  1.0000 &       0 &  1.0000 &  1.0000 &  1.0000 &  1.0000\\  \hline
		\end{tabular}
	\label{tab:p170}
\end{table}

\begin{table}[h]
	\centering
	\caption{The simulated Powers of the Test of independence based on several test statistics against the correlated random-walk type normal alternative with different correlation coefficients $r$ and for sample size $n=100$ }
		\begin{tabular}{|c| c c c c c c c c c|} \hline
		r    &  $T_1$ & $T_2$ & $T_3$ &  $T_4$  &  $T_5$ & $T_6$ &  $Spearman's~~ \hat{\rho}$ & $Kendall's~~ t$ & $T_p$ \\ \hline
    -1.0000 &       0 &       0 &  1.0000 &  1.0000 &       0 &  1.0000 &  1.0000 &  1.0000 &  1.0000\\
   -0.9000 &  0.0169 &  0.8378 &  0.9454 &  0.9569 &  0.1806 &  0.9351 &  0.9912 &  0.9927 &  0.9922\\
   -0.8000 &  0.0191 &  0.8458 &  0.8571 &  0.8829 &  0.1823 &  0.8401 &  0.9596 &  0.9644 &  0.9634\\
   -0.7000 &  0.0303 &  0.7335 &  0.7780 &  0.8117 &  0.1581 &  0.7564 &  0.9200 &  0.9220 &  0.9252\\
   -0.6000 &  0.0338 &  0.5634 &  0.7101 &  0.7490 &  0.1474 &  0.6820 &  0.8698 &  0.8706 &  0.8799\\
   -0.5000 &  0.0429 &  0.3711 &  0.6550 &  0.7035 &  0.1328 &  0.6267 &  0.8400 &  0.8404 &  0.8507\\
   -0.4000 &  0.0558 &  0.2442 &  0.6198 &  0.6659 &  0.1315 &  0.5891 &  0.8076 &  0.8022 &  0.8178\\
   -0.3000 &  0.0668 &  0.1479 &  0.5913 &  0.6398 &  0.1264 &  0.5564 &  0.7705 &  0.7628 &  0.7836\\
   -0.2000 &  0.0798 &  0.0956 &  0.5714 &  0.6225 &  0.1214 &  0.5411 &  0.7554 &  0.7466 &  0.7711\\
   -0.1000 &  0.0988 &  0.0634 &  0.5531 &  0.6011 &  0.1142 &  0.5236 &  0.7478 &  0.7376 &  0.7646\\
         0 &  0.1258 &  0.0574 &  0.5558 &  0.6072 &  0.1142 &  0.5240 &  0.7446 &  0.7359 &  0.7642\\
    0.1000 &  0.1580 &  0.0643 &  0.5532 &  0.6006 &  0.1139 &  0.5206 &  0.7501 &  0.7401 &  0.7666\\
    0.2000 &  0.1915 &  0.0942 &  0.5692 &  0.6208 &  0.1169 &  0.5368 &  0.7622 &  0.7537 &  0.7763\\
    0.3000 &  0.2351 &  0.1525 &  0.5890 &  0.6370 &  0.1240 &  0.5575 &  0.7795 &  0.7717 &  0.7972\\
    0.4000 &  0.2894 &  0.2390 &  0.6150 &  0.6619 &  0.1237 &  0.5805 &  0.8107 &  0.8069 &  0.8220\\
    0.5000 &  0.3562 &  0.3796 &  0.6455 &  0.6915 &  0.1308 &  0.6157 &  0.8339 &  0.8358 &  0.8491\\
    0.6000 &  0.4394 &  0.5612 &  0.7064 &  0.7439 &  0.1418 &  0.6759 &  0.8787 &  0.8829 &  0.8894\\
    0.7000 &  0.5448 &  0.7437 &  0.7713 &  0.8041 &  0.1494 &  0.7475 &  0.9161 &  0.9208 &  0.9215\\
    0.8000 &  0.6734 &  0.8424 &  0.8548 &  0.8779 &  0.1683 &  0.8386 &  0.9585 &  0.9617 &  0.9610\\
    0.9000 &  0.8414 &  0.8394 &  0.9465 &  0.9577 &  0.1862 &  0.9383 &  0.9929 &  0.9937 &  0.9937\\
    1.0000 &  1.0000 &       0 &  1.0000 &  1.0000 &       0 &  1.0000 &  1.0000 &  1.0000 &  1.0000\\   \hline
		\end{tabular}
	\label{tab:p1100}
\end{table}

\begin{table}[h]
	\centering
	\caption{The simulated Powers of the Test of independence based on several test statistics against the correlated random-walk type normal alternative with different correlation coefficients $r$ and for sample size $n=150$ }
		\begin{tabular}{|c| c c c c c c c c c|} \hline
		r    &  $T_1$ & $T_2$ & $T_3$ &  $T_4$  &  $T_5$ & $T_6$ &  $Spearman's~~ \hat{\rho}$ & $Kendall's~~ t$ & $T_p$ \\ \hline
   -1.0000 &       0 &       0 &  1.0000 &  1.0000 &       0 &  1.0000 &  1.0000 &  1.0000 &  1.0000\\
   -0.9000 &  0.0178 &  0.9339 &  0.9679 &  0.9760 &  0.2784 &  0.9626 &  0.9954 &  0.9967 &  0.9963\\
   -0.8000 &  0.0223 &  0.9542 &  0.9083 &  0.9270 &  0.2393 &  0.8929 &  0.9662 &  0.9690 &  0.9675\\
   -0.7000 &  0.0270 &  0.9024 &  0.8519 &  0.8808 &  0.2151 &  0.8295 &  0.9339 &  0.9376 &  0.9374\\
   -0.6000 &  0.0404 &  0.7600 &  0.8027 &  0.8357 &  0.1886 &  0.7750 &  0.9037 &  0.9047 &  0.9060\\
   -0.5000 &  0.0447 &  0.5538 &  0.7509 &  0.7928 &  0.1738 &  0.7224 &  0.8639 &  0.8634 &  0.8713\\
   -0.4000 &  0.0549 &  0.3558 &  0.7235 &  0.7654 &  0.1633 &  0.6887 &  0.8390 &  0.8327 &  0.8463\\
   -0.3000 &  0.0692 &  0.2037 &  0.7052 &  0.7500 &  0.1612 &  0.6717 &  0.8223 &  0.8147 &  0.8359\\
   -0.2000 &  0.0835 &  0.1211 &  0.6811 &  0.7316 &  0.1473 &  0.6396 &  0.8051 &  0.8009 &  0.8150\\
   -0.1000 &  0.1050 &  0.0765 &  0.6737 &  0.7246 &  0.1460 &  0.6351 &  0.7943 &  0.7879 &  0.8110\\
         0 &  0.1284 &  0.0532 &  0.6721 &  0.7231 &  0.1439 &  0.6346 &  0.7973 &  0.7929 &  0.8112\\
    0.1000 &  0.1565 &  0.0796 &  0.6750 &  0.7261 &  0.1474 &  0.6382 &  0.7950 &  0.7875 &  0.8114\\
    0.2000 &  0.2022 &  0.1168 &  0.6802 &  0.7303 &  0.1452 &  0.6449 &  0.7989 &  0.7914 &  0.8189\\
    0.3000 &  0.2351 &  0.2148 &  0.6999 &  0.7474 &  0.1581 &  0.6648 &  0.8174 &  0.8098 &  0.8318\\
    0.4000 &  0.3005 &  0.3629 &  0.7219 &  0.7680 &  0.1635 &  0.6857 &  0.8366 &  0.8341 &  0.8479\\
    0.5000 &  0.3624 &  0.5572 &  0.7623 &  0.8035 &  0.1727 &  0.7287 &  0.8676 &  0.8649 &  0.8771\\
    0.6000 &  0.4466 &  0.7607 &  0.7991 &  0.8375 &  0.1918 &  0.7705 &  0.9008 &  0.9022 &  0.9033\\
    0.7000 &  0.5522 &  0.9051 &  0.8492 &  0.8773 &  0.2106 &  0.8283 &  0.9329 &  0.9353 &  0.9366\\
    0.8000 &  0.6884 &  0.9589 &  0.9084 &  0.9292 &  0.2487 &  0.8930 &  0.9677 &  0.9697 &  0.9689\\
    0.9000 &  0.8400 &  0.9337 &  0.9683 &  0.9752 &  0.2775 &  0.9601 &  0.9939 &  0.9947 &  0.9945\\
    1.0000 &  1.0000 &       0 &  1.0000 &  1.0000 &       0 &  1.0000 &  1.0000 &  1.0000 &  1.0000\\  \hline
		\end{tabular}
	\label{tab:p1150}
\end{table}

\begin{table}[h]
	\centering
	\caption{The simulated Powers of the Test of independence based on several test statistics against the correlated random-walk type normal alternative with different correlation coefficients $r$ and for sample size $n=200$ }
		\begin{tabular}{|c| c c c c c c c c c|} \hline
		r    &  $T_1$ & $T_2$ & $T_3$ &  $T_4$  &  $T_5$ & $T_6$ &  $Spearman's~~ \hat{\rho}$ & $Kendall's~~ t$ & $T_p$ \\ \hline
   -1.0000 &       0 &       0 &  1.0000 &  1.0000 &       0 &  1.0000 &  1.0000 &  1.0000 &  1.0000\\
   -0.9000 &  0.0174 &  0.9733 &  0.9794 &  0.9856 &  0.3935 &  0.9740 &  0.9967 &  0.9966 &  0.9961\\
   -0.8000 &  0.0233 &  0.9888 &  0.9320 &  0.9472 &  0.3169 &  0.9187 &  0.9719 &  0.9730 &  0.9735\\
   -0.7000 &  0.0283 &  0.9718 &  0.8915 &  0.9171 &  0.2754 &  0.8693 &  0.9440 &  0.9469 &  0.9475\\
   -0.6000 &  0.0337 &  0.8784 &  0.8493 &  0.8834 &  0.2468 &  0.8237 &  0.9157 &  0.9160 &  0.9219\\
   -0.5000 &  0.0448 &  0.6986 &  0.8159 &  0.8559 &  0.2170 &  0.7796 &  0.8847 &  0.8839 &  0.8915\\
   -0.4000 &  0.0539 &  0.4664 &  0.7868 &  0.8339 &  0.2056 &  0.7498 &  0.8629 &  0.8621 &  0.8725\\
   -0.3000 &  0.0681 &  0.2783 &  0.7707 &  0.8201 &  0.1917 &  0.7351 &  0.8467 &  0.8394 &  0.8560\\
   -0.2000 &  0.0765 &  0.1443 &  0.7509 &  0.7972 &  0.1782 &  0.7100 &  0.8299 &  0.8258 &  0.8418\\
   -0.1000 &  0.1049 &  0.0860 &  0.7521 &  0.7992 &  0.1754 &  0.7098 &  0.8268 &  0.8185 &  0.8388\\
         0 &  0.1310 &  0.0635 &  0.7415 &  0.7904 &  0.1732 &  0.7043 &  0.8246 &  0.8142 &  0.8352\\
    0.1000 &  0.1590 &  0.0820 &  0.7486 &  0.7979 &  0.1767 &  0.7114 &  0.8243 &  0.8137 &  0.8378\\
    0.2000 &  0.1964 &  0.1544 &  0.7544 &  0.7990 &  0.1779 &  0.7143 &  0.8336 &  0.8249 &  0.8473\\
    0.3000 &  0.2411 &  0.2651 &  0.7695 &  0.8132 &  0.1888 &  0.7355 &  0.8426 &  0.8356 &  0.8560\\
    0.4000 &  0.3001 &  0.4641 &  0.7893 &  0.8310 &  0.1969 &  0.7538 &  0.8628 &  0.8609 &  0.8716\\
    0.5000 &  0.3596 &  0.6934 &  0.8149 &  0.8571 &  0.2219 &  0.7837 &  0.8783 &  0.8767 &  0.8847\\
    0.6000 &  0.4576 &  0.8844 &  0.8512 &  0.8874 &  0.2403 &  0.8246 &  0.9162 &  0.9197 &  0.9222\\
    0.7000 &  0.5500 &  0.9701 &  0.8945 &  0.9180 &  0.2724 &  0.8721 &  0.9472 &  0.9496 &  0.9495\\
    0.8000 &  0.6953 &  0.9872 &  0.9353 &  0.9517 &  0.3164 &  0.9197 &  0.9728 &  0.9751 &  0.9743\\
    0.9000 &  0.8496 &  0.9749 &  0.9789 &  0.9849 &  0.3886 &  0.9743 &  0.9954 &  0.9959 &  0.9961\\
    1.0000 &  1.0000 &       0 &  1.0000 &  1.0000 &       0 &  1.0000 &  1.0000 &  1.0000 &  1.0000\\   \hline
		\end{tabular}
	\label{tab:p1200}
\end{table}

\begin{figure}[h]
\centering
\includegraphics[width=79mm, height=60mm] {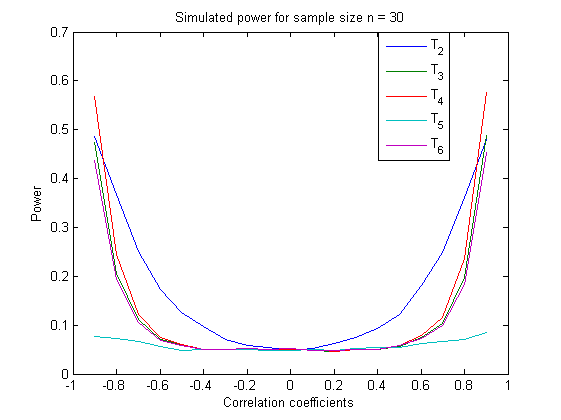}
\includegraphics[width=79mm, height=60mm] {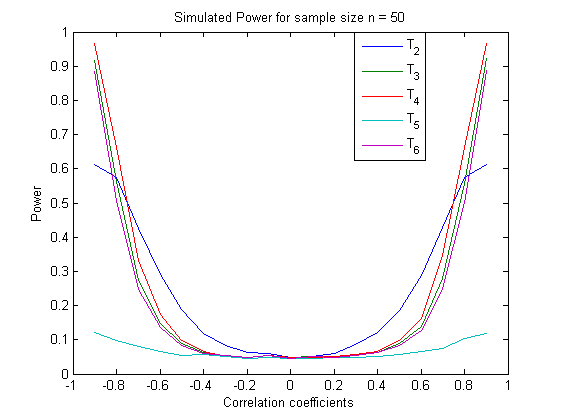}
\includegraphics[width=79mm, height=60mm] {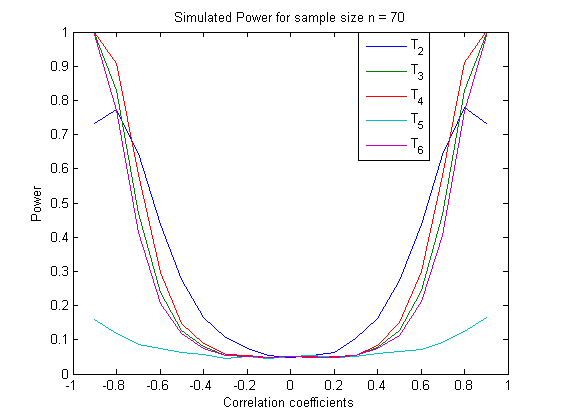}
\includegraphics[width=79mm, height=60mm] {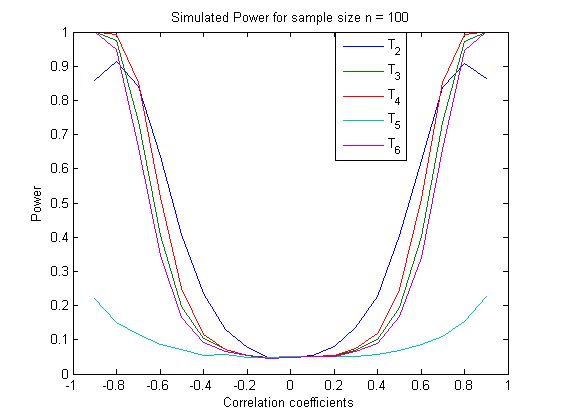}
\includegraphics[width=79mm, height=60mm] {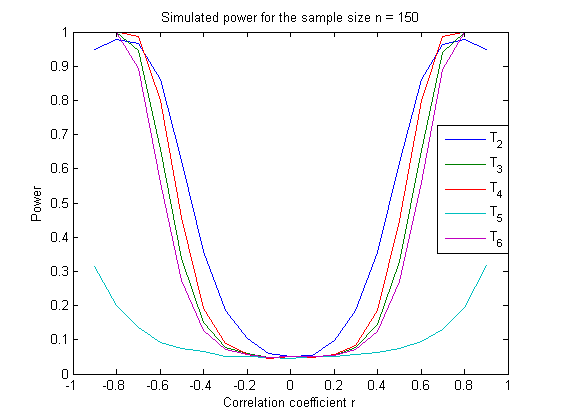}
\includegraphics[width=79mm, height=60mm] {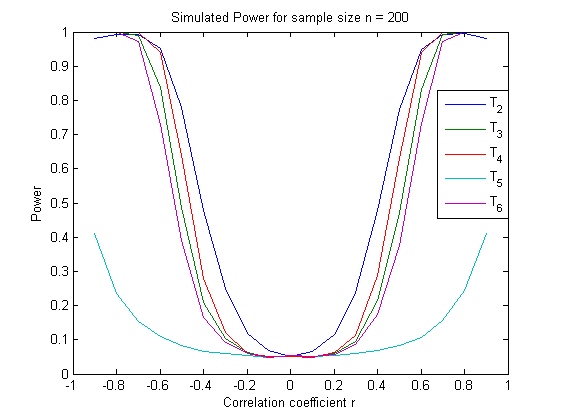}
\caption{Plot of the Simulated power aginst the correlated normal alternative for different tests for various sample sizes $n$}
\label{fig:p1}
\end{figure}

\begin{figure}[h]
\centering
\includegraphics[width=79mm, height=60mm] {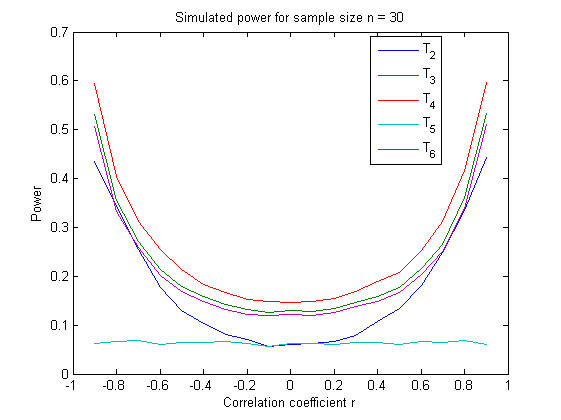}
\includegraphics[width=79mm, height=60mm] {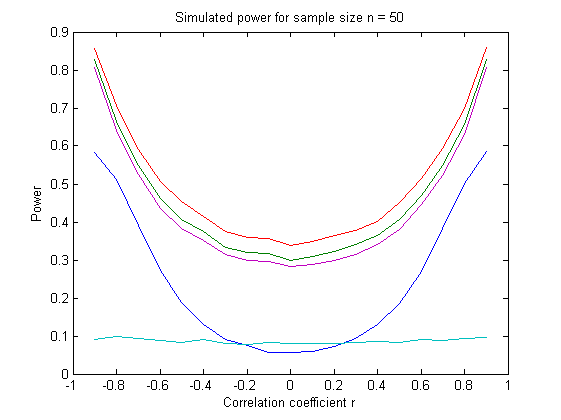}
\includegraphics[width=79mm, height=60mm] {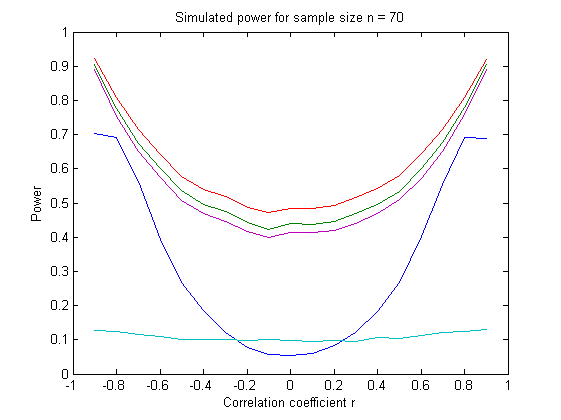}
\includegraphics[width=79mm, height=60mm] {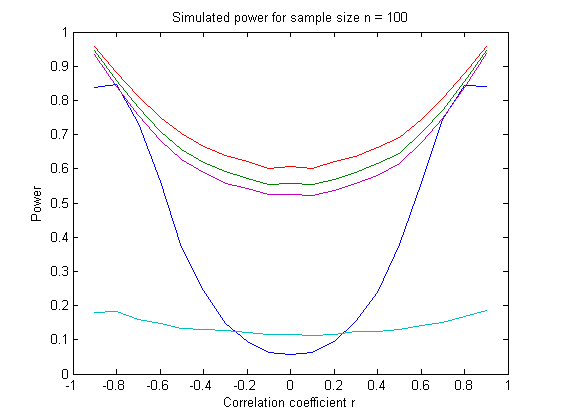}
\includegraphics[width=79mm, height=60mm] {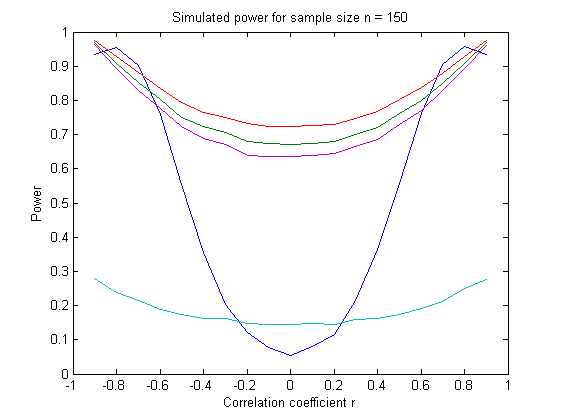}
\includegraphics[width=79mm, height=60mm] {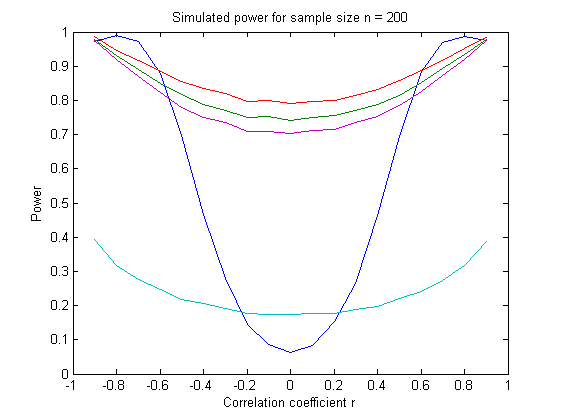}
\caption{Plot of the Simulated power aginst the the correlated random-walk type normal alternative for different tests for sample sizes $n= 30 $ and $n=50$}
\label{fig:p2}
\end{figure}

\end{document}